 \documentclass[accepted]{uai2021} % after acceptance, for a revised
\usepackage[american]{babel}
\usepackage{natbib} % has a nice set of citation styles and commands
\usepackage{mathtools} % amsmath with fixes and additions
\usepackage{booktabs} % commands to create good-looking tables
\usepackage{tikz} % nice language for creating drawings and diagrams

\usepackage{amsthm}
\usepackage{amsmath}
\usepackage[ruled,vlined]{algorithm2e}
\usepackage{amssymb}
\usepackage{subcaption}
\usepackage{tabularx}

\theoremstyle{proposition}
\newtheorem{proposition}{Proposition}

\theoremstyle{remark}

\theoremstyle{definition}

\newcommand{\N}{\mathcal{N}}

\newcommand{\z}{\mathbf{z}}

\newcommand{\I}{\mathbf{I}}
\newcommand{\y}{\mathbf{y}}
\newcommand{\sigset}{\{\sigma_{j,m}^2\}_{j,m}}

\newcommand{\var}{\mathbf{P}}
\newcommand{\K}{\mathbf{K}}
\newcommand{\realz}{\mathbf{z}^{\Re}}
\newcommand{\imz}{\mathbf{z}^{\Im}}
\newcommand{\pow}{\psi}
\newcommand{\Pow}{\Psi}

\newcommand{\zvec}{\widetilde{\z}}
\newcommand{\E}{\mathbb{E}}
\newcommand{\first}{f_{1}(\Psi)}
\newcommand{\second}{f_{2}(\Psi)}

\DeclareMathOperator*{\argmin}{arg\,min}

\title{PLSO: A generative framework for decomposing nonstationary time-series into piecewise stationary oscillatory components}

\author[1]{Andrew H. Song} % Lead author
\author[2]{Demba Ba}
\author[3,4,5]{Emery N. Brown}
\affil[1]{%
    Electrical Engineering and Computer Science\\
    Massachusetts Institute of Technology\\
    Cambridge, Massachusetts, USA
}
\affil[2]{%
    School of Engineering and Applied Sciences\\
    Havard University\\
    Cambridge, Massachusetts, USA
}
\affil[3]{
	Picower Institute of Learning and Memory\\
	Massachusetts Institute of Technology\\
	Cambridge, Massachusetts, USA
}
\affil[4]{
	Institute of Medical Engineering and Sciences\\
	Massachusetts Institute of Technology\\
	Cambridge, Massachusetts, USA
}
\affil[5]{
	Department of Anesthesia, Critical Care and Pain Medicine\\
	Massachusetts General Hospital\\
	Boston, Massachusetts, USA
}

\begin{document}
\maketitle

\begin{abstract}
  To capture the slowly time-varying spectral content of real-world time-series, a common paradigm is to partition the data into approximately stationary intervals and perform inference in the time-frequency domain.  However, this approach lacks a corresponding nonstationary time-domain generative model for the entire data and thus, time-domain inference occurs in each interval separately. This results in distortion/discontinuity around interval  boundaries  and can consequently lead to erroneous inferences based on  any  quantities derived from the posterior, such as the phase. To  address  these  shortcomings,  we propose the Piecewise Locally Stationary Oscillation (PLSO) model  for  decomposing   time-series data  with  slowly time-varying  spectra into several oscillatory, piecewise-stationary  processes. PLSO,  as  a  nonstationary  time-domain  generative model, enables inference on the entire time-series without boundary  effects and simultaneously provides a characterization of its time-varying spectral properties. We also propose a novel two-stage inference algorithm that combines Kalman theory and an accelerated  proximal  gradient  algorithm.  We demonstrate  these  points  through  experiments  on simulated data and real neural data from the rat and the human brain.
\end{abstract}

\section{Introduction}
With the collection of long time-series now common, in areas such as neuroscience and geophysics, it is important to develop an \textit{inference} framework for data where the stationarity assumption is too restrictive. We restrict our attention to data 1) with spectral properties that change slowly over time and 2) for which decomposition into several oscillatory components is warranted for interpretation, often the case in electroencephalogram (EEG) or electrophysiology recordings. One can use bandpass filtering~\citep{Oppenheim2009} or methods such as the empirical mode decomposition~\citep{Huang98, daubechies11} for these purposes. However, due to the \textit{absence} of a generative model, these methods lack a framework for performing inference. Another popular approach is to perform inference in the time-frequency (TF) domain on the short-time Fourier transform (STFT) of the data, assuming stationarity within small intervals. This has led to a rich literature on inference in the TF domain, such as~\citet{Wilson08}. A drawback is that most of these methods focus on estimates for the power spectral density (PSD) and lose important phase information. To recover the time-domain estimates, additional algorithms are required~\citep{Griffin84}.

This motivates us to explore time-domain generative models that allow time-domain inference and decomposition into oscillatory components. We can find examples based on the stationarity assumption in the signal processing/Gaussian process (GP) communities. A superposition of stochastic harmonic oscillators, where each oscillator corresponds to a frequency band, is used in the processing of speech ~\citep{Cemgil2005} and neuroscience data~\citep{Matsuda17, Beck18}. In GP literature \citep{GP}, the spectral mixture (SM) kernel~\citep{Wilson2013, Wilkinson19} models the data as samples from a GP, whose kernel consists of the superposition of localized and frequency-modulated kernels. 

These time-domain models can be applied to nonstationary data by partitioning them into stationary intervals and performing time-domain inference within each interval. However, we are faced with a different kind of challenge. As the inference is localized within each interval, the time-domain estimates in different intervals are independent conditioned on the data and do not reflect the dependence across intervals. This also causes discontinuity/distortion of the time-domain estimates near the interval boundaries, and consequently any quantities derived from these estimates.

To address these shortcomings, we propose a generative framework for data with slow time-varying spectra, termed the Piecewise Locally Stationary Oscillation (PLSO) framework\footnote{Code is available at https://github.com/andrewsong90/plso.git}. The main contributions are:

\textbf{Generative model for piecewise stationary, oscillatory components} PLSO models time-series as the superposition of piecewise stationary, oscillatory components. This allows time-domain inference on each component and estimation of the time-varying spectra.

\textbf{Continuity across stationary intervals} The state-space model that underlies PLSO strikes a balance between ensuring time-domain continuity across piecewise stationary intervals and stationarity within each interval. Moreover, by imposing stochastic continuity on the interval-level, PLSO learns underlying smooth time-varying spectra accurately. 

\textbf{Inference procedure} We propose a two-stage inference procedure for the time-varying spectra and the time-series. By leveraging the Markovian dynamics, the algorithm combines Kalman filter theory~\citep{Kalman60} and inexact accelerated proximal gradient approach~\citep{Li15}.

In Section~\ref{section:background} we introduce necessary background, followed by the PLSO framework in Section~\ref{section:PSLO}. In Section~\ref{section:inference}, we discuss inference for PLSO. In Section~\ref{section:connection}, we discuss how PLSO relates to other frameworks. In Section~\ref{section:experiments}, we present experimental results and conclude in Section~\ref{section:concolusion}.
%\footnote{Code: https://github.com/ds2p/plso.git.}.

\section{Background }\label{section:background}
\subsection{Notation}
We use $j\in\{1, \ldots, J\}$ and $k\in\{1, \dots, K \}$ to denote frequency and discrete-time sample index, respectively. 
We use $\omega\in[-\pi,\pi]$ for normalized frequency. The $j^{\text{th}}$ latent process centered at $\omega=\omega_j$ is denoted as $\z_j\in\mathbb{C}^K$, with $\z_{j,k}\in\mathbb{C}$ denoting the $k^{\text{th}}$ sample of $\z_{j}$ and $\realz_{j,k}$, $\imz_{j,k}$ its real and imaginary parts. We also represent $\z_{j,k}$ as a $\mathbb{R}^2$ vector, $\zvec_{j,k}=[\realz_{j,k}, \imz_{j,k}]^{\text{T}}$. The elements of $\z_j$ are denoted as $\z_{j,k:k'}=[\z_{j,k},\ldots,\z_{j,k'}]^{\text{T}}$.  The state covariance matrix for $\z_{j,k}$ is defined as $\var_{k}^j=\E[\zvec_{j,k}\left(\zvec_{j,k}\right)^{\text{T}} ]$. To express an enumeration of variables, we use $\{\cdot\}$ and drop first/last index for simplicity, e.g. $\{\z_j\}_j$ instead of $\{\z_j\}_{j=1}^J$. 

We use $\y_k$ and $\z_{j,k}$ for the discrete-time counterpart of the continuous observation and latent process, $y(t)$ and $z_j(t)$. With the sampling frequency $f_s=1/\Delta$, we have $\y_k=y(k\Delta)$, $\z_{j,k}=z_j(k\Delta)$, and $T=K\Delta$.

\subsection{Piecewise local stationarity}
The concept of \textit{piecewise local stationarity} (PLS) for nonstationary time-series with slowly time-varying spectra~\citep{Adak98} plays an important role in PLSO. A stationary process has a constant mean and a covariance function which depends only on the difference between two time points.

For our purposes, it suffices to understand the following on PLS: 1) It includes local stationary~\citep{Priestley65,dahlhaus1997} and amplitude-modulated stationary processes. 2) A PLS process can be approximated as a piecewise stationary (PS) process (Theorem 1 of~\citep{Adak98})
%First, a stochastic process is locally stationary (LS) at a specific time point $u$, if there exists an interval centered at $u$ within which the stationary can be assumed. A process is PLS, if it is LS at all time points (except for a finite number of them). More importantly, a PLS process can be approximated by a sequence of piecewise stationary (PS) processes (Theorem 1 of~\cite{Adak98})
\begin{equation}\label{eq:ps_adak}
z(t) = \sum_{m=1}^{M}\mathbf{1}(u_{m}\leq t < u_{m+1})\cdot z^{m}(t),
\end{equation}
where $z^{m}(t)$ is a continuous stationary process and the boundaries are $0= u_1 < \cdots < u_{M+1}= T$. Note that Eq.~\ref{eq:ps_adak} does not guarantee continuity across different PS intervals,
\begin{equation*}
\begin{split}
\lim_{t\rightarrow u^{-}_{m}}z(t)=\lim_{t\rightarrow u^{-}_{m}}z^{m-1}(t)\neq\lim_{t\rightarrow u^{+}_{m}}z^{m}(t)=\lim_{t\rightarrow u^{+}_{m}}z(t).\\
\end{split}
\end{equation*}

\section{The PLSO model and its mathematical properties}\label{section:PSLO}

Building on the Theorem 1 of~\citep{Adak98}, PLSO models nonstationary data as PS processes. It is a superposition of $J$ different PS processes $\{\z_j \}_j$, with $\z_j$ corresponding to an oscillatory process centered at frequency $\omega_j$. PLSO also guarantees stochastic continuity across PS intervals. We show that piecewise stationarity and continuity across PS intervals are two competing objectives and that PLSO strikes a balance between them, as discussed in Section~\ref{section:PLSO_nonstationary}. 

\begin{figure}[!ht]
	\centering
	\includegraphics[width=0.93\linewidth]{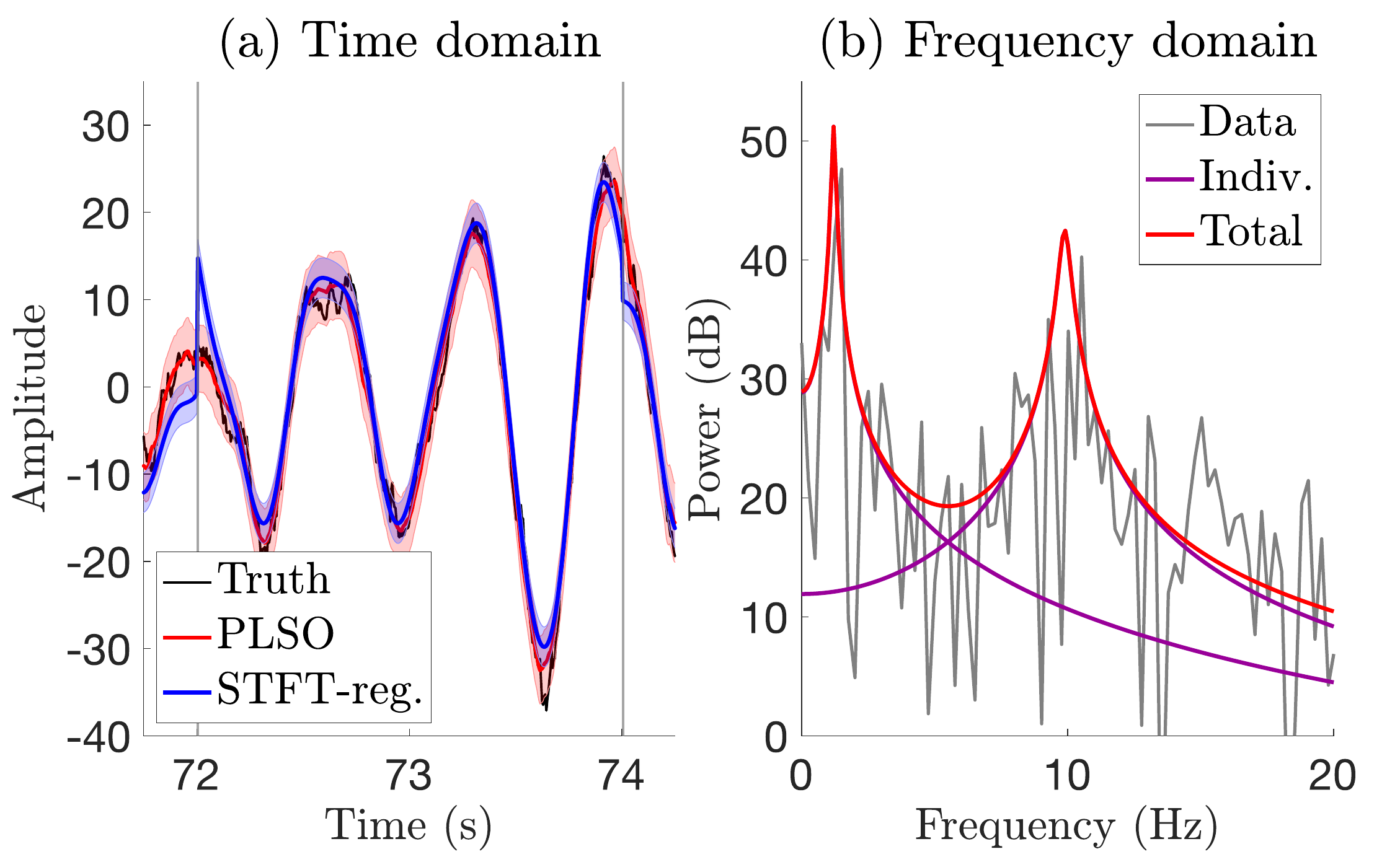}
	\caption{A simulated example. (a) Time domain. Data (black) around boundaries (gray) and inferred oscillatory components using PLSO (red) and regularized STFT (blue). (b) Frequency domain. Spectrum of the data (gray), PLSO components for $J=2$ (purple) and their sum (red).}
	\label{fig:example}
\end{figure}

Fig.~\ref{fig:example} shows an example of the PLSO framework applied to simulated data. In the time domain, the oscillation inferred using the regularized STFT (blue)~\citep{Kim18}, which imposes stochastic continuity on the STFT coefficients, suffers from discontinuity/distortion near window boundaries, whereas that inferred by PLSO (red) does not. In the frequency domain, each PLSO component corresponds to a localized spectrum $S_j(\omega)$, the sum of which is the PSD $\gamma(\omega)$, and is fit to the data STFT (or periodogram) $I(\omega)$. We start by introducing the PLSO model for a single window.

\subsection{PLSO for stationary data} As a building block for PLSO, we use the discrete stochastic harmonic oscillator for a stationary time series~\citep{Qi2002, Cemgil2005, Matsuda17}. The data $\y$ are assumed to be a superposition of $J$ \textit{independent} zero-mean components
\begin{equation}\label{eq:PSLO_stationary}
\begin{split}
\zvec_{j,k}&= \rho_j\mathbf{R}(\omega_j)\zvec_{j,k-1}+\varepsilon_{j, k}\ \\
\y_k&=\sum_{j=1}^J\realz_{j,k} +\nu_k, \\
\end{split}
\end{equation}
where $\mathbf{R}(\omega_j)=\begin{pmatrix}
\cos(\omega_j) & -\sin(\omega_j)\\
\sin(\omega_j) & \cos(\omega_j)\\
\end{pmatrix}$, $\varepsilon_{j,k}\sim \mathcal{N}\left(0,\alpha_j\I_{2\times 2}\right)$, and $\nu_{k}\sim \mathcal{N}(0,\sigma_{\nu}^2)$, correspond to the rotation matrix, the state noise, and  the observation noise, respectively. The imaginary component $\imz_{j,k}$, which does not directly contribute to $\y_k$, can be seen as the auxiliary variable to write $\z_j$ in recursive form using $\mathbf{R}(\omega_j)$~\citep{Koopman09}. We assume $\var_{1}^j=\sigma_j^2\cdot\mathbf{I}_{2\times 2}\,\,\,\forall j$. 

We reparameterize $\rho_j$ and $\alpha_j$, using lengthscale $l_j$ and power $\sigma^2_j$, such that $\rho_j=\exp(-\Delta/l_j)$ and $\alpha_j = \sigma_j^2 (1 - \rho_j^2)$. Theoretically, this establishes a connection to 1) the stochastic differential equation~\citep{Brown04, Solin14}, detailed in \textbf{Appendix A}, and 2) a superposition of frequency-modulated and localized GP kernels, similar to SM kernel~\citep{Wilson2013}. Practically, this ensures that $\rho_j<1$, and thus stability of the process.

Given Eq.~\ref{eq:PSLO_stationary}, we can readily express the frequency spectra of PLSO in each interval, through the autocovariance function. The autocovariance of $\z_j$ is given as $Q_j(n')=\E[\realz_{j,k}\realz_{j,k+n'} ] = \sigma_j^2\cos(\omega_j n')\exp\left(-n'\Delta /l_j\right)$. It can also be thought of as an \textit{exponential} kernel, frequency-modulated by $\omega_j$. The spectra for $\z_j$, denoted as $S_j(\omega)$, is obtained by taking the Fourier transform (FT) of $Q_j(n')$
\begin{equation*}
\begin{split}
S_j(\omega)&=\sum_{n'=-\infty}^{\infty}Q_j(n')\exp\left(-i\omega n'\right)=\varphi_j(\omega)+\varphi_j(-\omega)\\
\varphi_j(\omega)&=\frac{\sigma_j^2\left(1 - \exp\left(-2\Delta/l_j\right)\right)}{1+\exp\left(-2\Delta/l_j\right)-2\exp\left(-\Delta/l_j\right)\cos(\omega-\omega_j)},\\
\end{split}
\end{equation*}
with the detailed derivation in \textbf{Appendix B}.

Given $S_j(\omega)$, we can show that PSD $\gamma(\omega)$ of the entire process $\sum_{j=1}^J\z_j$ is simply $\gamma(\omega)=\sum_{j=1}^JS_j(\omega)$. First, since $\z_j$ is independent across $j$, the autocovariance can be simplified, i.e., $\E[\sum_j\realz_{j,k}\sum_j\realz_{j,k+n'} ]=\sum_j\E[\realz_{j,k}\realz_{j,k+n'} ]$. Next, using the linearity of FT, we can conclude that $\gamma(\omega)$ is a superposition of individual spectra.

\subsection{PLSO for nonstationary data}\label{section:PLSO_nonstationary}
If $\y$ is nonstationary, we can still apply stationary PLSO of Eq.~\ref{eq:PSLO_stationary} for the time-domain inference. However, this implies constant spectra for the entire data ($S_j(\omega)$ and $\gamma(\omega)$ do not depend on $k$), which is not suitable for nonstationary time-series for which we want to track spectral dynamics. This point is further illustrated in Section~\ref{section:experiments}.

We therefore segment $\y$ into $M$ non-overlapping PS intervals, indexed by $m\in\{1,\ldots, M\}$, of length $N$, indexed by $n\in\{1,\ldots, N\}$, such that $K=MN$. We then apply the stationary PLSO to each interval, with additional Markovian dynamics imposed on $\sigma_{j,m}^2$,
\begin{equation}\label{eq:tfgpss2}
\begin{split}
\log(\sigma_{j,m}^2)&=\log(\sigma_{j,m-1}^2)+\eta_{j,m}\\
\zvec_{j,mN+n}&= \rho_j\mathbf{R}(\omega_j)\zvec_{j,mN+(n-1)}+\varepsilon_{j, mN+n}\\
\y_{mN+n}&=\sum_{j=1}^J\realz_{j,mN+n}+\nu_{mN+n},\\
\end{split}
\end{equation}
where $\varepsilon_{j, mN+n}\sim\N(0, \sigma_{j,m}^2(1-\rho_j^2)\I_{2\times 2})$, $ \eta_{j,m}\sim \N(0,1/\sqrt{\lambda})$ and $\nu_{mN+n}\sim \N(0,\sigma_{\nu}^2)$. We define $\var_{m,n}^j$ as the covariance of $\zvec_{j,mN+n}$, with $\var_{1,1}^j = \sigma_{j,1}^2\I_{2\times 2},\forall j$. The graphical model is shown in Fig.~\ref{fig:graphical_model}.

\begin{figure}[!ht]
	\centering
	\includegraphics[width=0.9\linewidth]{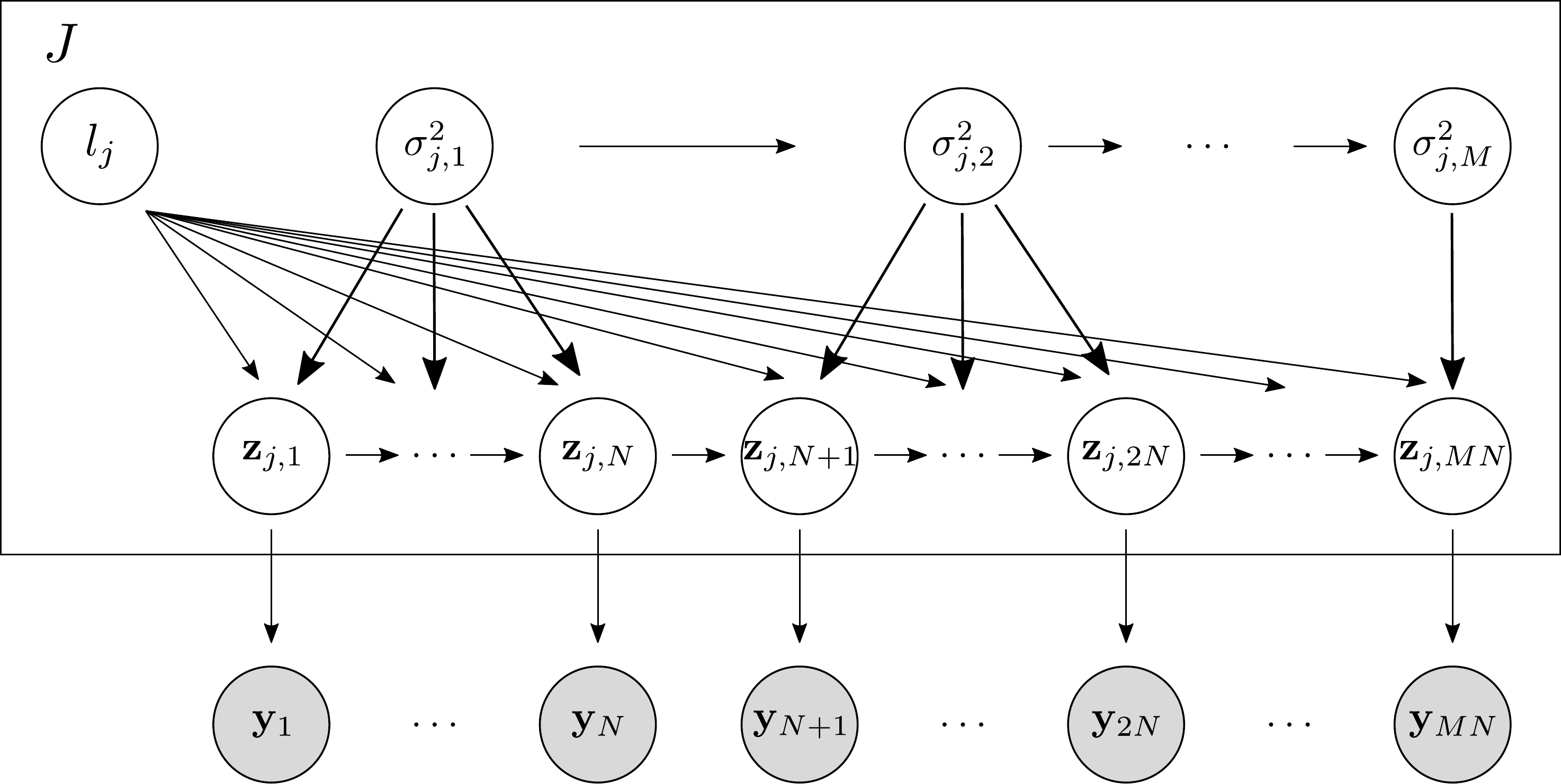}
	\caption{The graphical model for PLSO.}
	\label{fig:graphical_model}
\end{figure}

We can understand PLSO as providing a \textit{parameterized} spectrogram defined by $\theta=\{\lambda,\sigma_{\nu}^2, \{l_j\}_j, \{\omega_j\}_j \}$ and $\sigset$ of the time-domain generative model. The lengthscale $l_j$ controls the bandwidth of the $j^{\text{th}}$ process, with \textit{larger} $l_j$ corresponding to  \textit{narrower} bandwidth.  The variance $\sigma_{j,m}^2$ controls the power of $\z_j$ and changes across different intervals, resulting in time-varying spectra $S_j^{(m)}(\omega)$ and PSD $\gamma^{(m)}(\omega)$. The center frequency $\omega_j$, and $-\omega_j$, at which $S_j^{(m)}(\omega)$ is maximized, controls the modulation frequency.

%The variance $\sigma_{j,m}^2$ changes across different intervals and results in time-varying spectra $S_j^{(m)}(\omega)$ and PSD $\gamma^{(m)}(\omega)$. We assume constant lengthscale $l_j$ to ensure that differences in $S_j^{(m)}(\omega)$ are caused only by differences in $\{\sigma^2_{j,m} \}_{j,m}$. \textcolor{blue}{}

%We first analyze a few important properties of the PLSO model.
As discussed previously, the segmentation approach for nonstationary time-series produces distortion/discontinuity artifacts around interval boundaries - The PLSO, as described by Eq.~\ref{eq:tfgpss2}, resolves these issues gracefully. We now analyze two mathematical properties, \textit{stochastic continuity} and \textit{piecewise stationarity}, to gain more insights on how PLSO accomplishes this.

\subsubsection{Stochastic continuity}
We discuss two types of stochastic continuity, 1) across the interval boundaries and 2) on $\{\sigma_{j,m}^2\}$.

\textbf{Continuity across the interval boundaries} In PLSO, the state-space model (Eq.~\ref{eq:tfgpss2}) provides stochastic continuity across different PS intervals. The following proposition rigorously explains stochastic continuity for PLSO.
\begin{proposition}\label{prop:continuity}
	For a given $m$, as $\Delta\rightarrow 0$, the samples on either side of the interval boundary, which are $\zvec_{j, (m+1)N}$ and $\zvec_{j, (m+1)N+1}$, converge to each other in mean square,
	\begin{equation*}
	\lim_{\Delta\rightarrow 0}\E[\Delta \zvec_{j, (m+1)N}\Delta \zvec_{j, (m+1)N}^{\text{T}}]=0,
	\end{equation*}
	where we use $\Delta \zvec_{j, (m+1)N}=\zvec_{j, (m+1)N+1}-\zvec_{j, (m+1)N}$.
\end{proposition}
\begin{proof}We use the connection between PLSO, which is a discrete-time model, and its continuous-time counterpart. Details are in the \textbf{Appendix C}.
\end{proof}
This matches our intuition that as $\Delta\rightarrow 0$, the adjacent samples from the same process should coverge to each other. For PS approaches without the sample-level continuity, even with the interval-level constraint \citep{Rosen09, Kim18, Song18, Das18, Soulat19}, convergence is not guaranteed.

We can interpret the continuity in the context of posterior for $\z_j$. For PS approaches without continuity, we have
\begin{equation}\label{eq:conditional_posterior}
\begin{split}
&p(\left\{\z_j\right\}_j|\y)\\
&\propto \prod_{m=1}^M p(\left\{\z_{j,(m-1)N+1:mN}\right\}_j | \y_{(m-1)N+1:mN}),\\
%&p\left(\left\{\z_j\right\}_j|\y\right)\propto p(\y|\{\z_j\}_j)\cdot p(\{\z_j\}_j)\\
%&=\prod_{m=1}^M p\left(\y_{(m-1)N+1:mN} |\left\{\z_{j,(m-1)N+1:mN}\right\}_j\right)\\
%&\quad\quad\quad \cdot p\left(\{\z_{j,(m-1)N+1:mN}\}_j\right)\\
%&\propto\prod_{m=1}^M p\left(\left\{\z_{j,(m-1)N+1:mN}\right\}_j | \y_{(m-1)N+1:mN}\right),\\
\end{split}
\end{equation}
where $\sigset$ and $\theta$ are omitted for notational ease. This is due to $p(\{\z_j \}_j)=\prod_{m=1}^M p(\{\z_{j,(m-1)N+1:mN}\}_j)$, as a result of \textit{absence} of continuity across the intervals. Consequently, the inferred time-domain estimates are conditionally independent across the intervals. On the contrary, in PLSO, the time-domain estimates depend on the entire $\y$, not just on a subset. 

\textbf{Continuity on $\sigma_{j,m}^2$} For a given $j$, we impose stochastic continuity on $\log \sigma_{j,m}^2$. Effectively, this pools together estimates of $\{\sigma^2_{j,m}\}_{m}$ to 1) prevent overfitting to the noisy data spectra and 2) estimate smooth dynamics of $\{\sigma^2_{j,m}\}_{m}$. The use of $\log \sigma_{j,m}^2$ ensures that $\sigma_{j,m}^2$ is non-negative. 

The choice of $\lambda$ dictates the smoothness of $\{\sigma^2_{j,m}\}_m$, with the two extremes corresponding to the familiar dynamics. If $\lambda\rightarrow 0$, we treat each window  independently. If $\lambda\rightarrow \infty$, we treat the data as stationary, as the constraint forces $\sigma^2_{j,m}=\sigma_j^2$, $\forall m$. Practically, the smooth regularization prevents artifacts in the spectral analysis, arising from sudden motion or missing data, as demonstrated in Section~\ref{section:experiments}.

%We take advantage of the state-space dynamics when optimizing $\sigset$, leading to a more intuitive/efficient optimization, analyzed in Section~\ref{section:inference}. 

\subsubsection{Piecewise stationarity}
For the $m^{\text{th}}$ window to be piecewise stationary, the initial state covariance matrix $\var_{m,1}^j$ should be the \textit{steady-state} covariance matrix for the window, denoted as $\var_{m,\infty}^{j}$. 
%This ensures the covariance is stationary within each window.

The challenge is transitioning from $\var_{m,\infty}^{j}$ to $\var_{m+1,\infty}^{j}$. 
%If these were equal, a stationary model would apply to the data. 
Specifically, to ensure $\var_{m+1,1}^{j}=\var_{m+1,\infty}^{j}$, given that $\var_{m,N}^{j}=\var_{m,\infty}^{j}\neq\var_{m+1,\infty}^{j}$, the \textit{variance }of the process noise between the two samples,  $\varepsilon_{j,(m+1)N+1}$, has to equal $\var_{m+1, \infty}^{j}-\exp\left(-2\Delta/l_j\right)\var_{m,\infty}^{j}$. However, this is infeasible. If $\var_{m+1,\infty}^{j}<\var_{m,\infty}^{j}$, the variance is negative. Even if it were positive, the limit as $\Delta\rightarrow 0$ does not equal zero, i.e., $\var_{m+1,\infty}^{j}-\var_{m,\infty}^{j}$. As a result, the Proposition~\ref{prop:continuity} no longer holds and the trajectory is discontinuous. 

In summary, there exists a \textit{trade-off} between maintaining piecewise stationarity and continuity across intervals. PLSO maintains continuity across the intervals while ensuring that the state covariance quickly transitions to the steady-state covariance. We quantify the speed of transition in the following proposition.

\begin{proposition}
	Assume $l_j \ll N\Delta$, such that $\var_{m,N}^j=\var_{m,\infty}^{j}$. In Eq.~\ref{eq:tfgpss2}, the difference between $\var_{m,\infty}^{j}$ and $\var_{m+1,\infty}^{j}$ decays exponentially fast as a function of $n$, 
%	for $1\leq n \leq N$,
	\begin{equation*}\label{eq:steady_state_induction}
	\begin{split}
	\var_{m+1,n}^j&=\var_{m+1,\infty}^{j} + \exp(-\frac{2n\Delta}{l_j})( \var_{m,\infty}^{j}-\var_{m+1,\infty}^{j}).\\
	\end{split}
	\end{equation*}
\end{proposition}
\begin{proof} In \textbf{Appendix D}, we prove this result by induction.
	
\end{proof}
This implies that, except for the transition portion at the beginning of each window, we can assume stationarity. In practice, we additionally impose an upper bound on $l_j$ during estimation and also use a reasonably-large  $N$. Empirically, we observe that the transition period has little impact.

%\subsection{Connection to STFT}
%For further insights, we examine the connections between our framework and STFT. Both treat the nonstationary data as piecewise stationary. The main difference, apart from the continuity across the windows, is the oscillatory components for each framework. In STFT, each component corresponds to a harmonic basis. In PLSO, each component is quasi-periodic which allows capturing of broader spectral content.
%
%By setting $l_j\rightarrow \infty$ in PLSO, we recover harmonic basis. We have $\lim_{l_j\rightarrow\infty} \exp(-\Delta/l_j)\mathbf{R}(\omega_j)=\mathbf{R}(\omega_j)$ and $\lim_{l_j\rightarrow\infty}  \E\left[\varepsilon_{j,mN+n}\varepsilon_{j,mN+n}^{\text{T}}\right] =0$, from Eq.~\ref{eq:tfgpss2}. Therefore, the transition becomes a determinstic rotation by $\omega_j$. 
%%This is different from the limit $\Delta\rightarrow 0$ in Proposition~\ref{prop:continuity}, where we use the equivalent continuous model. It is still a discrete model for $\l_j\rightarrow \infty$. 
%Consequently, we have $S_j^{(m)}(w) = \left(\var_{j,m}^{\infty}/2\right)\left\{\mathbf{1}(w=w_j) + \mathbf{1}(w=-w_j)\right\}$ in the frequency domain. This agrees with our intuition that as $\l_j\rightarrow \infty$, the $S_j^{(m)}(w)$ has zero bandwidth and is equivalent to delta functions placed at $\{\omega_j\}_j$.

\section{Inference}\label{section:inference}
Given the generative model in Eq.~\ref{eq:tfgpss2}, our goal is to perform inference on the posterior distribution
\begin{equation}\label{eq:posterior}
\begin{split}
&p(\left\{\z_j\right\}_j,\sigset\mid\y, \theta)\\
&=\underbrace{p(\sigset\mid\y,\theta)}_{\text{window-level posterior}}\cdot \underbrace{p(\left\{\z_j\right\}_j\mid\sigset,\y,\theta)}_{\text{sample-level posterior}}.\\
\end{split}
\end{equation}
We can learn $\theta$ or fix the parameters to specific values informed by domain knowledge, such as the center frequency or bandwidth of the processes. The posterior distribution factorizes into two terms as in Eq.~\ref{eq:posterior}, the \textit{window-level} posterior $p(\sigset|\y,\theta)$ and the \textit{sample-level} posterior $p(\{\z_j\}_j|\sigset,\y,\theta)$. Accordingly, we break the inference into two stages. 

\textbf{Stage 1} We minimize the \textit{window-level} negative log-posterior, with respect to $\theta$ and $\sigset$. Specifically, we obtain maximum likelihood (ML) estimate $\hat{\theta}_{\text{ML}}$ and maximum a posteriori (MAP) estimate $\{\widehat{\sigma}^2_{j,m, \text{MAP}}\}_{j,m}$. We drop subscripts ML and MAP for notational simplicity.

\textbf{Stage 2} Given $\{\widehat{\sigma}^2_{j,m}\}_{j,m}$ and $\hat{\theta}$, we perform inference on the \textit{sample-level} posterior. This includes computing the mean $\widehat{\z}_j=\E[\z_j| \{\widehat{\sigma}^2_{j,m}\}_{j,m}, \y, \hat{\theta} ]$ and credible intervals, which quantifies the uncertainty of the estimates, or any statistical quantity derived from the posterior distribution.

\subsection{Optimization of $\sigset$ and $\theta$}
We factorize the posterior, $p(\sigset\mid\y,\theta)\propto p(\y \mid \sigset,\theta )\cdot p(\sigset\mid \theta )$. As the exact inference is intractable, we instead minimize the negative log-posterior, $-\log p(\sigset\mid\y,\theta)$. This is an \textit{empirical Bayes} approach~\citep{Casella85}, since we estimate $\sigset$ using the marginal likelihood $p(\y \mid \sigset,\theta )$. The smooth hyperprior provides the MAP estimate for $\sigset$.

We use the \textit{Whittle likelihood}~\citep{whittle1953}, defined for stationary time-series in the frequency domain, for the log-likelihood $f(\sigset;\theta)=\log p(\y\mid \sigset,\theta)$,
\begin{equation}
\begin{split}
&f(\sigset;\theta)\\
&=-\frac{1}{2}\sum_{m,n=1}^{M, N}\log(\gamma^{(m)}(\omega_n)+\sigma_{\nu}^2)+\frac{I^{(m)}(\omega_n)}{\gamma^{(m)}(\omega_n)+\sigma_{\nu}^2},\\
\end{split}
\end{equation}
where the log-likelihood is the sum of the Whittle likelihood computed for each interval, with discrete frequency $\omega_n=2\pi n/N,$ and data STFT $I^{(m)}(\omega_n)=\left|\sum_{n'=1}^N\exp\left(-2\pi i(n'-1)(n-1)/N\right)\y_{mN+n'}\right|^2.$
The Whittle likelihood, which is nonconvex, enables frequency-domain parameter estimation as a computationally more efficient alternative to the time domain estimation \citep{Turner2014}. The concave log-prior $g(\sigset;\theta)=\log p(\sigset\mid\theta)$, which arises from the continuity on $\sigset$, is given as
\begin{equation}\label{eq:logprior}
g(\sigset;\theta)=-\frac{\lambda}{2}\sum_{j=1}^J\sum_{m=1}^M \left(\log \sigma^2_{j,m} - \log \sigma^2_{j,m-1} \right)^2.
\end{equation}
This yields the following nonconvex problem
\begin{equation}\label{eq:penalized_whittle}
\begin{split}
&\min_{\sigset,\theta}\, -\log p\left( \sigset \mid \y,\theta\right)\\
&=\min_{\sigset,\theta}\, -f(\sigset;\theta) - g(\sigset;\theta). \\
\end{split}
\end{equation}

We optimize Eq.~\ref{eq:penalized_whittle} by block coordinate descent~\citep{Wright15} on $\sigset$ and $\{\sigma_{\nu}^2,\{l_j\}_j, \{\omega_j\}_j\}$. For $\sigma_{\nu}^2$, $\{l_j\}_j$, and $\{\omega_j\}_j$, we minimize $-f(\sigset;\theta)$, since $g(\sigset;\theta)$ does not affect them. 
We perform conjugate gradient descent on $\{l_j\}_j$ and $\{\omega_j\}_j$. 
We discuss initialization and how to estimate $\sigma_{\nu}^2$ in the \textbf{Appendix E}.

\subsubsection{Optimization of $\sigset$}
We introduce an algorithm to compute a local optimal solution of $\sigset$ for the nonconvex optimization problem in Eq.~\ref{eq:penalized_whittle}, by leveraging the regularized temporal structure of $\sigset$. It extends the inexact accelerated proximal gradient (APG) method~\citep{Li15}, by solving the proximal step with a Kalman filter/smoother~\citep{Kalman60}. This follows since computing the proximal operator for $g(\sigset;\theta)$ is equivalent to MAP estimation for $J$ independent 1-dimensional linear Gaussian state-space models
\begin{equation}\label{eq:proximal}
\begin{split}
&\{\log \sigma^{(l+1), 2}_{j,m}\}_{j,m}=\operatorname{prox}_{-\alpha^{(l)} g}(\mathbf{v}^{(l)})\\
&=\argmin_{\sigset}\underbrace{\frac{\sum_{j,m}^{J,M}(\mathbf{v}_{j,m}^{(l)} - \log \sigma^2_{j,m})^2}{2\alpha^{(l)}}-g(\sigset)}_{\sum_{j=1}^J q_j}\\
\end{split}
\end{equation}
where $\alpha^{(l)}>0$ is a step-size for the $l^{\text{th}}$ iteration, $q_j=\sum_{m=1}^M \frac{(\mathbf{v}^{(l)}_{j,m}- \log \sigma^2_{j,m})^2}{2\alpha^{(l)}}+\frac{\lambda}{2}(\log\sigma^2_{j,m}-\log\sigma^2_{j,m-1})^2$, and $\mathbf{v}^{(l)}_{j,m}=\log \sigma_{j,m}^{(l),2}+\alpha^{(l)}\frac{\partial f(\sigset)}{\partial\log\sigma^2_{j,m}}\Big|_{\{\sigma^{(l), 2}_{j,m}\}_{j,m}}$.
%\begin{equation*}
%\begin{split}
%\mathbf{v}^{(l)}_{j,m}&=\log \sigma_{j,m}^{(l),2}+\alpha^{(l)}\frac{\partial f(\sigset)}{\partial\log\sigma^2_{j,m}}\Bigg|_{\{\sigma^{(l), 2}_{j,m}\}_{j,m}}.\\
%\end{split}
%\end{equation*}

The $j^{\text{th}}$ optimization problem, $\min_{\{\sigma_{j,m}^2\}_m} q_j$, is equivalent to estimating the mean of the posterior for $\{\log \sigma^2_{j,m}\}_m$ in a linear Gaussian state-space model with observations $\{\mathbf{v}^{(l)}_{j,m}\}_m$, observation noise variance $\alpha^{(l)}$, and state variance $1/\lambda$. Therefore, the solution can efficiently be computed with $J$ 1-dimensional, Kalman filters/smoothers, with the computational complexity of $O(JM)$.

Note that Eq.~\ref{eq:proximal} holds for all non-negative $\lambda$. If $\lambda=0$, the proximal operator is an identity operator, as $\log \sigma^{(l+1),2}_{j,m}=\mathbf{v}^{(l)}_{j,m}$. This is a gradient descent with a step-size rule. If $\lambda\rightarrow \infty$, we have $\log \sigma^2_{j,m}=\log \sigma^2_{j,m-1}$, $\forall m$, which leads to $\log\sigma_{j,m}^{(l+1),2}=(1/M)\sum_{m=1}^M\mathbf{v}^{(l)}_{j,m}$.
%\textcolor{red}{We can also replace $g(\Pow;\theta)$ with more complicated priors, such as ones based on mixed norms~\cite{kowalski09, Ba14} to promote group sparsity. This is possible since the algorithm iteratively approximates the complicated Whittle likelihood with a simpler/well-understood Gaussian likelihood, as in Eq.~\ref{eq:proximal}.}
The algorithm is guaranteed to converge when $\alpha^{(l)}< 1/C$, where $C$ is the Lipschitz constant for $f(\sigset;\theta)$.  In practice, we select $\alpha^{(l)}$ according to the step-size rule~\citep{barzilai88}. In \textbf{Appendix F \& G}, we present the full algorithm for optimizing $\sigset$ and a derivation for $C$.

%Instead, we propose an alternate scheme by replacing $\lambda_j$ with $\lambda\beta_j$, where $\beta_j$ is some heuristically determined weighting factor and $\sum_{j=1}^J\beta_j=1$. With $\beta_j$ determined automatically from the data, the computational complexity of the grid search reduces drastically. We propose the inverse of the variance of $\{\widehat{\pow}_{j,m}\}_m$ estimated from the independent approach as $\beta_j$, i.e., $\beta_j\propto \left(\operatorname{var}\left(\{\pow_{j,m}\}_m\right)\right)^{-1}$. Similar weighted regularization approaches can be found in other context, such as penalized weighted least squares~\cite{Wang06} and group/adaptive lasso~\cite{Zou06}.

\subsection{Inference for $p(\left\{\z_j\right\}_j\mid\sigset,\y,\theta)$}
We perform inference on the posterior distribution $p(\left\{\z_j\right\}_j\mid  \{\widehat{\sigma}_{j,m}^2\}_{j,m}, \y,\hat{\theta})$. Since this is a Gaussian distribution, the mean trajectories $\{\widehat{\z}_j\}_j$ and the credible intervals can be computed analytically. Moreover, Eq.~\ref{eq:tfgpss2} is a linear Gaussian state-space model, we can use Kalman filter/smoother for efficient computation with computational complexity $O(J^2K)$, further discussed In \textbf{Appendix H}. Since we use the point estimate $\{\widehat{\sigma}^2_{j,m}\}_{j,m}$, the credible interval for $\{\widehat{\z}_j\}_j$ will be \textit{narrower} compared to the full Bayesian setting which accounts for all values of $\sigset$.
%\begin{equation*}
%\begin{split}
%&p(\left\{\z_j\right\}_j\mid\y, \hat{\theta})\\
%&=\int p(\left\{\z_j\right\}_j\mid\sigset,\y, \hat{\theta})p(\sigset\mid\y, \hat{\theta})d\sigset.\\
%\end{split}
%\end{equation*}

\subsubsection{Monte Carlo Inference}
We can also obtain posterior samples and perform Monte Carlo (MC) inference on any posterior-derived quantity. To generate the MC trajectory samples, we use the forward-filter backward-sampling (FFBS) algorithm~\citep{Carter94}. Assuming $S$ number of MC samples, the computational complexity for FFBS is $O(SJ^2K)$, since for each sample, the algorithm uses Kalman filter/smoother for sampling. This is different from generating samples with the interval-specific posterior in Eq.~\ref{eq:conditional_posterior}. In the latter case, the FFBS algorithm is run $M$ times, the samples of which have to be concatenated to form an entire trajectory. With PLSO, the trajectory sample is conditioned on the entire observation and is continuous across the intervals.

One quantity of interest is the \textit{phase}. We obtain the phase as $\phi_{j,k}=\tan^{-1}(\imz_{j,k}/\realz_{j,k})$. Since $\tan^{-1}(\cdot)$ is a non-linear operation, we compute the mean and credible interval with MC samples through the FFBS algorithm. Given the posterior-sampled trajectories $\{\z_{j}^{(s)}\}_{j,s}$, where $s\in\{1,\ldots,S\}$ denotes MC sample  index, we estimate $\widehat{\phi}_{j,k}=(1/S)\sum_{s=1}^S\tan^{-1}(\z^{\Im, (s)}_{j,k}/\z^{\Re, (s)}_{j,k})$, and use empirical quantiles for the associated credible interval. 

\subsection{Choice of $J$ and $\lambda$}
We choose $J$ that minimizes the Akaike Information Criterion (AIC) \citep{akaike81}, defined as 
\begin{equation}
\operatorname{AIC}(J) = -(2/M)\cdot\log p( \y \mid \{\widehat{\sigma}_{j,m}^2\}_{j,m}, \hat{\theta}) +2\cdot 3\cdot J, 
\end{equation}	
where $3\cdot J$ corresponds to the number of parameters ($\{l_j\}_j, \{\omega_j\}_j, \{\sigma^2_{j,m}\}_j$). The regularization parameter $\lambda$ is determined through a two-fold cross-validation, where each fold is generated by aggregating even/odd sample indices \citep{Ba14}.

\subsection{Choice of window length $N$}
The choice of window length $N$ presents the tradeoff between 1) spectral resolution and 2) the temporal resolution of the spectral dynamics~\citep{Oppenheim2009}. For a shorter window, the estimated spectral dynamics have a \textit{finer} temporal resolution, \textit{coarser} spectral resolution, and \textit{higher} variance. For a longer window, these trends are reversed. This suggests that the choice is application-dependent. For electrophysiology data, a window on the order of seconds is used, as scientific interpretations are made on the basis of fine spectral resolution (< 1Hz). For audio signal processing~\citep{Gold11}, short windows ($10\sim 100$ ms) are used, since audio data is highly nonstationary and thus requires fine temporal resolution for processing. A survey of window lengths used in different applications can be found in the \textbf{supporting information} of \citet{Kim18}.

%\textcolor{blue}{We focused on equal-length windows, as these are predominantly used in neuroscience experiments and hospitals. We note that PLSO can easily be run on top of the segmentations produced by the automatic segmentation or with variable window lengths, with minimal change to the inference procedure to account for different window lengths.}

\section{Related works}\label{section:connection}
We examine how PLSO relates to other nonstationary frameworks.

\textbf{STFT/Regularized STFT} In STFT, the harmonic basis is used, whereas quasi-periodic components are used for PLSO, which allows capturing of broader spectral content. Recent works regularize STFT coefficients with stochastic continuity across the windows, to infer smooth spectral dynamics \citep{Ba14, Kim18}. However, this regularization leads to discontinuities at window boundaries.

\textbf{Piecewise stationary GP} GP regression and parameter estimation are performed within each interval \citep{Gramacy08, Solin18}. Consequently, the recovered trajectories are discontinuous. Also, the inversion of covariance matrix leads to an expensive inference. For example, the time-complexity of mean trajectory estimation is $O(MN^3)=O(N^2K)$, whereas the time-complexity for PLSO is $O(J^2K)$. Considering that the typical sampling frequency for electrophysiology data is $\sim 10^3$ (Hz) and windows are several seconds, which leads to $N\sim 10^3$, PLSO is computationally more efficient. In the \textbf{Appendix I}, we confirm this through an experiment.

\textbf{Time-varying Autoregressive model (TVAR)} 
The TVAR model~\citep{kitagawa85} is given as
\begin{equation*}
\y_k=\sum_{p=1}^P a_{p,k}\y_{k-p}+\varepsilon_k,
\end{equation*}	
with the time-varying coefficients $\{a_{p,k}\}_p$. Consequently, it does not suffer from discontinuity issues. TVAR can also be \textit{approximately} decomposed into oscillatory components via eigen-decomposition of the transition matrix \citep{West1999}. 
However, since the eigen-decomposition changes at every sample, this leads to an ambiguous interpretation of the oscillations in the data, as we discuss in Section~\ref{section:experiments}.

\textbf{RNN frameworks}
Despite the popularity of recurrent neural networks (RNN) for time-series applications~\citep{Goodfellow16}, we believe PLSO is more appropriate for time-frequency analysis for two reasons.
\begin{enumerate}
	\item
	RNNs operate in the time-domain with the goal of prediction/denoising and consequently less emphasis is placed on local stationarity or estimation of second-order statistics. Performing time-frequency analysis requires segmenting the RNN outputs and applying the STFT, which yields noisy spectral estimates.
	\item RNN is not a generative framework. Although variational framework can be combined with RNN~\citep{Chung15,Krishnan17}, the use of variational lower bound objective could lead to suboptimal results. On the other hand, PLSO is a generative framework that maximizes the true log-posterior.
\end{enumerate}

\section{Experiments}~\label{section:experiments}

We apply PLSO to three settings: 1) A simulated dataset 2) local-field potential (LFP) data from the rat hippocampus, and 3) EEG data from a subject under anesthesia.

We use PLSO with $\lambda=0$, $\lambda\rightarrow\infty$, and $\lambda$ determined by cross-validation, $\lambda_{\operatorname{CV}}$. We use interval lengths chosen by domain experts. As baselines, we use 1) regularized STFT (STFT-reg.) and 2) Piecewise stationary GP (GP-PS). For GP-PS, we use the same $\{\widehat{\sigma}_{j,m}^2\}_{j,m}$ and $\widehat{\theta}$ as PLSO with $\lambda=0$, so that the estimated PSD of GP-PS and PLSO are identical. This lets us explain differences in time-domain estimates by the fact that PLSO operates in the time-domain.

\subsection{Simulated dataset}

We simulate from the following model for $1\leq k \leq K$
\begin{equation*}
\y_k = 10\left(\frac{K-k}{K} \right)\realz_{1,k} + 10\cos^4(2\pi \omega_0 k)\realz_{2,k}+\nu_k,
\end{equation*}
where $\z_{1,k}$ and $\z_{2,k}$ are as in Eq.~\ref{eq:PSLO_stationary}, with $(\omega_0 ,\omega_1,\omega_2)=(0.04,1,10)$ Hz, $f_s=200$ Hz, $T=100$ seconds, $l_1=l_2=1$, and $\nu_k\sim\N(0,25)$. This stationary process comprises two amplitude-modulated oscillations, namely one modulated by a slow-frequency sinusoid and the other a linearly-increasing signal~\citep{Ba14}. We simulate 20 realizations and train on each realization, assuming 2-second PS intervals. For PLSO, we use $J=2$. Additional details are provided in the \textbf{Appendix J}.

\textbf{Results} We use two metrics: 1) Mean squared error (MSE) between the mean estimate $\hat{\z}_j$ and the ground truth $\z_j^{\text{True}}$ and 2) $\operatorname{jump}(\z_j) $. The averaged results are shown in Table~\ref{table:sim_result}. We define $\operatorname{jump}(\z_j) = \frac{1}{M-1}\sum_{m=1}^{M-1} \left|\hat{\z}_{j, mN+1} -\hat{\z}_{j, mN}\right|$ to be the level of discontinuity at the interval boundaries. If $\operatorname{jump}(\z_j) $ greatly exceeds $\operatorname{jump}(\z_j^{\text{True}})$, this implies the existence of large discontinuities at the boundaries.

\begin{figure}[!h]
	\centering
	\includegraphics[width=\linewidth]{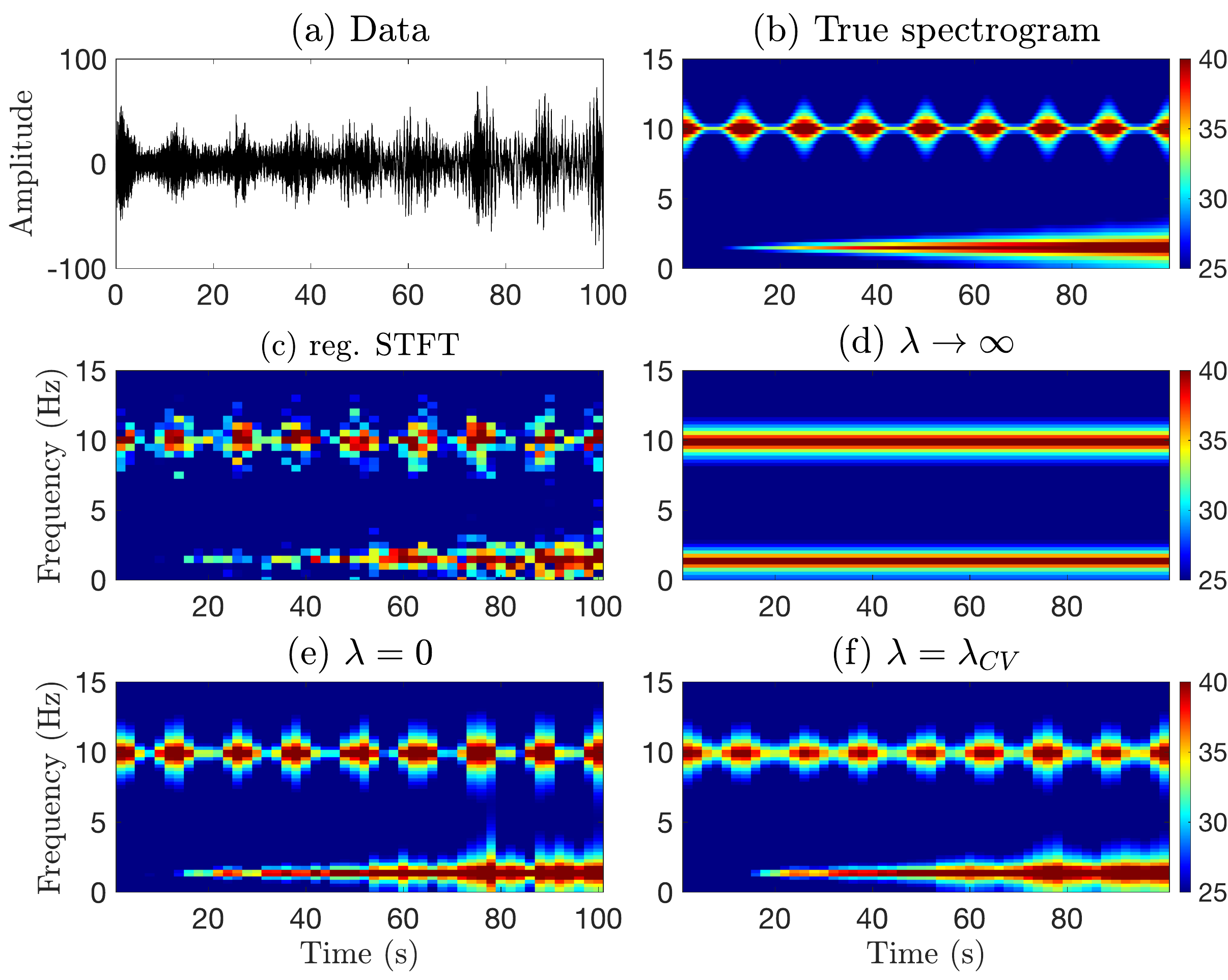}
	\caption{Spectrograms for simulation (in dB). (a) True data (b) True spectrogram (c) regularized STFT (d) PLSO with $\lambda\rightarrow \infty$ (e) PLSO with $\lambda=0$ (f) PLSO with $\lambda=\lambda_{\operatorname{CV}}$.}
	\label{fig:simulation}
\end{figure}

\begin{figure*}[ht!]
	\centering
	\includegraphics[width=0.95\linewidth]{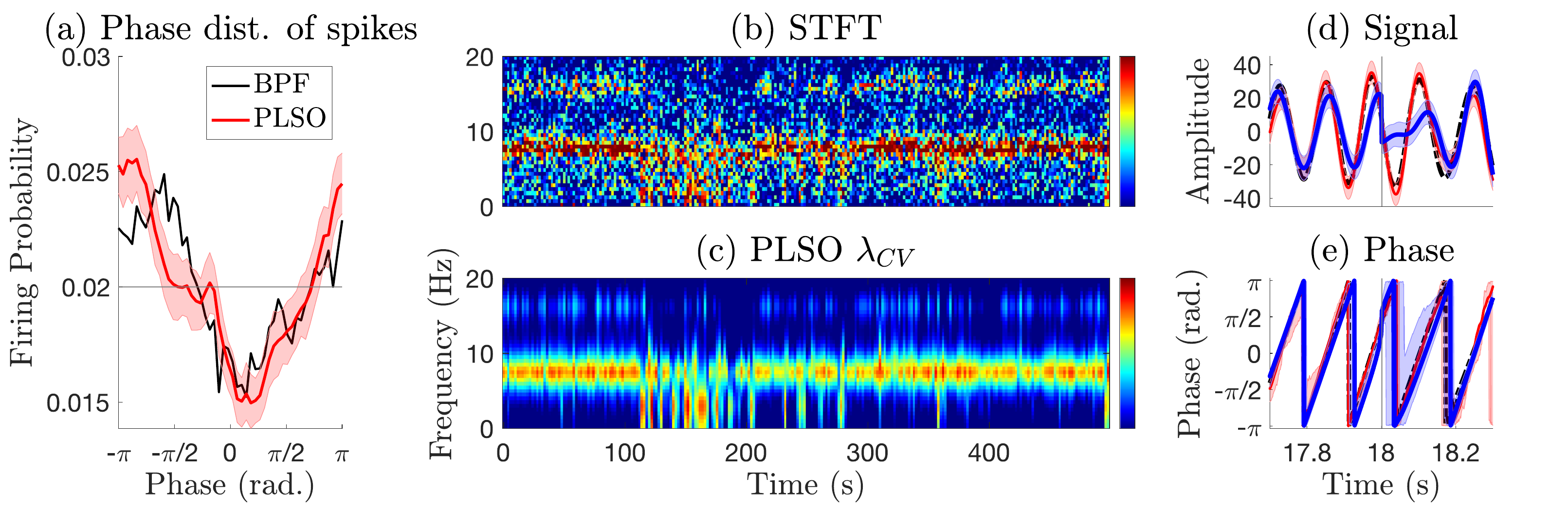}
	\caption{Result of analyses of hippocampal data. (a) Theta phase distribution of population neuron spikes, computed with bandpass-filtered LFP (black), PLSO estimate of $\hat{\z}_2$ with credible intervals estimated from 200 posterior samples (red). Horizontal gray line indicates the uniform distribution. (b-c) Spectrogram (in dB) for 500 seconds (b) STFT (c) PLSO with $\lambda_{\operatorname{CV}}$. Learned frequencies are $(\widehat{\omega}_1,\widehat{\omega}_2,\widehat{\omega}_3)=(2.99,7.62,15.92)$ Hz, with $\widehat{\omega}_4\sim\widehat{\omega}_5>25$ Hz. (d-e) Time-domain results. (d) Reconstructed signal (e) phase for $\widehat{\z}_2$ and interval boundary (vertical gray), with bandpass-filtered data (dotted black), STFT-reg. (blue), and PLSO (red). Shaded area represents 95\% credible interval from $S=200$ sample trajectories.}
	\label{fig:hc}
\end{figure*}

\begin{table}[h!]
	\caption{Simulation results. For $\operatorname{jump}(\z_j)$ and MSE, left/right metrics correspond to $\z_1$/$\z_2$, respectively.}
	\label{table:sim_result}
	\begin{center}
		\begin{tabular}{|c||c|c|c|c|}
			\hline
			& $\operatorname{jump}(\z_j)$ & MSE & IS div.\\
			\hline
			Truth & 0.95/12.11&0/0 & 0 \\
			$\lambda=0$ & 0.26/10.15& 2.90/3.92 & 4.08\\
			$\lambda\rightarrow \infty$ & 0.22/10.32&3.26/4.53&13.78\\
			$\lambda=\lambda_{\operatorname{CV}}$ & 0.25/10.21&\textbf{2.88}/\textbf{3.91}&\textbf{3.93}\\
			STFT-reg. & 49.59/81.00&6.89/10.68&N/A\\
			GP-PS & 16.99/23.28&3.00/4.04& 4.08\\
			\hline
		\end{tabular}
	\end{center}
\end{table}

Fig.~\ref{fig:simulation} shows the true data in the time domain and spectrogram results.  Fig.~\ref{fig:simulation}(c) shows that although the regularized STFT detects activities around $1$ and $10$ Hz, it fails to delineate the time-varying spectral pattern. Fig.~\ref{fig:simulation}(d) shows that PLSO with stationarity ($\lambda\rightarrow \infty$) assumption is too restrictive. Fig.~\ref{fig:simulation}(e), (f) show that both PLSO with independent window assumption ($\lambda=0$) and PLSO with cross-validated $\lambda$ ($\lambda=\lambda_{\operatorname{CV}}$) are able to capture the dynamic pattern, with the latter being more effective in recovering the smooth dynamics across different PS intervals.

For GP-PS and STFT-reg., $\operatorname{jump}(\z_j)$ exceeds $\operatorname{jump}(\z_j^{\text{True}})$, indicating discontinuities at the boundaries. An example is in Fig.~\ref{fig:example}. PLSO produces a similar jump metric as the ground truth metric, indicating the absence of discontinuities. We attribute the lower value to Kalman smoothing. For the TF domain, we use Itakura-Saito (IS) divergence~\citep{Itakura} as a distance measure between the ground truth spectra and the PLSO estimates. That the highest divergence is given by $\lambda\rightarrow\infty$ indicates the inaccuracy of the stationarity assumption. 

\subsection{LFP data from the rat hippocampus}
We use LFP data collected from the rat hippocampus during open field tasks~\citep{Mizuseki09}, with $T=1,600$ seconds and $f_s=1,250$ Hz\footnote{We use channel 1 of mouse ec013.528 for the LFP. The population spikes were simultaneously recorded.}. The theta neural oscillation band ($5\sim10$ Hz) is believed to play a role in coordinating the firing of neurons in the entorhinal-hippocampal system and is important for understanding the local circuit computation.

We fit PLSO with $J=5$, which minimizes AIC as shown in Table~\ref{table:hc}, with 2-second PS interval. The estimated $\widehat{\omega}_2$ is 7.62 Hz in the theta band. To obtain the phase for non-PLSO methods, we perform the Hilbert transform on the theta-band reconstructed signal. With no ground truth, we bandpass-filter (BPF) the data in the theta band for reference.

\begin{table}[!ht]
	\caption{AIC as a function of $J$ for Hippocampus data}
	\label{table:hc}
	\begin{center}
		\begin{tabular}{|c|c|c|c|c|c|c|}
			\hline
			J&1 & 2  & 3 & 4 & \textbf{5} & 6\\
			\hline
			AIC & 2882 & 2593 & 2566 & 2522 & \textbf{2503} & 2505\\
			\hline
		\end{tabular}
	\end{center}
\end{table}

\textbf{Spike-phase coupling} Fig.~\ref{fig:hc}(a) shows the theta phase distribution of population neuron spikes in the hippocampus. The PLSO-estimated distribution (red) confirms the original results analyzed with bandpass-filtered signal (black)~\citep{Mizuseki09}--the hippocampal spikes show a strong preference for a specific phase, $\pi$ for this dataset, of the theta band. Since PLSO provides posterior sample-trajectories for the entire time-series, we can compute as many realizations of the phase distribution as the number of MC samples. The resulting credible interval excludes the uniform distribution (horizontal gray), suggesting the statistical significance of strong phase preference. 

\textbf{Denoised spectrogram} Fig.~\ref{fig:hc}(b-c) shows the estimated spectrogram. We observe that PLSO denoises the spectrogram, while retaining sustained power at $\widehat{\omega}_2=7.62$ Hz and weaker bursts at $(\widehat{\omega}_1,\widehat{\omega}_3)= (2.99, 15.92)$ Hz. 
%PLSO with $\lambda_{\operatorname{CV}}$ captures the nonstationarity well.

\textbf{Time domain discontinuity} Fig.~\ref{fig:hc}(d-e) show a segment of the estimated signal and phase near a boundary for $\widehat{\omega}_2$. While the estimates from STFT-reg. (blue) and PLSO (red) follow the BPF result closely, the STFT-reg. estimates exhibit discontiunity/distortion near the boundary. In Fig.~\ref{fig:hc}(e), the phase jump at the boundary is $38.4$ degrees. We also computed $\operatorname{jump}(\phi_2)$ in degrees/sample. Considering that the theta band roughly progresses $2.16\,(= 7.5 (\text{Hz}) \times 360/1250\,(\text{Hz}))$ degrees/sample, we observe that BPF (2.23), as expected, and PLSO ($\lambda_{\operatorname{CV}}:$ 2.40, $\lambda\rightarrow\infty$: 2.66) are not affected by the boundary effect. This is not the case for STFT-reg. (26.83) and GP-PS (25.91).

%\begin{table}[!ht]
%	\caption{\textcolor{blue}{$\operatorname{jump}(\phi_2)$ for different approaches at boundaries.}}
%	\label{table:hc2}
%		\begin{center}
%			\begin{tabular}{|c|c|c|c|c|c|}
%				\hline
%				&$\lambda_{\operatorname{CV}}$ & $\lambda\rightarrow \infty$  & STFT-reg. & GP-PS& BPF\\
%				\hline
%				jump & 2.40 & 2.66 & 26.83 & 25.91 & 2.23 \\
%				\hline
%			\end{tabular}
%		\end{center}
%\end{table}

\begin{figure}[!ht]
	\centering
	\includegraphics[width=0.98\linewidth]{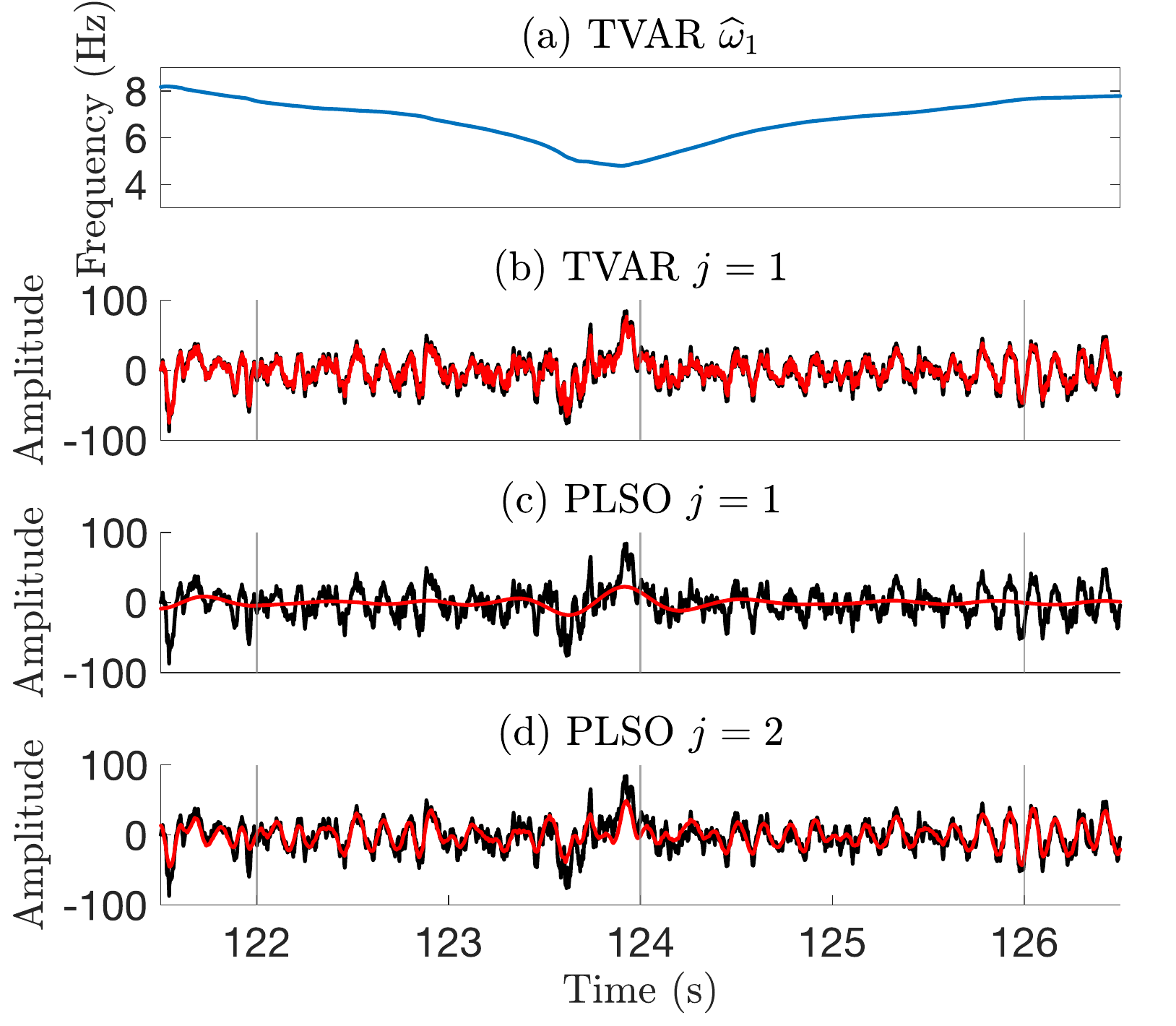}
	\caption{Hippocampus data. (a) time-varying $\widehat{\omega}_1$ for TVAR. (b-d) Inferred mean trajectory (red) for (b) TVAR $j=1$, (c) PLSO $j=1$, and (d) PLSO $j=2$, with raw data (black). }
	\label{fig:tvar}
\end{figure}

\textbf{Comparison with TVAR} Fig.~\ref{fig:tvar}(a-b) shows a segment of TVAR inference results\footnote{The details for TVAR is in the \textbf{Appendix K}.}. Specifically, Fig.~\ref{fig:tvar}(a) and (b) shows a time-varying center frequency $\widehat{\omega}_1$ and the corresponding reconstruction, for the lowest frequency component. Note that the eigenvalues, which correspond to $\{\omega_j\}_j$, and the eigenvectors, which are used for oscillatory decomposition, are derived from the \textit{estimated} TVAR transition matrix. Consequently, we cannot explicitly control $\{\omega_j\}_j$, as shown in Fig.~\ref{fig:tvar}(a), the bandwidth of each component, as well as the number of components $J$. This is further complicated by the fact that the transition matrix changes every sample.
Fig.~\ref{fig:tvar}(b) shows that this ambiguity results in the lowest-frequency component of TVAR explaining \textit{both} the slow ($0.1\sim 2$ Hz) and theta components. With PLSO, on the contrary, we can explicitly specify or learn the parameters. Fig.~\ref{fig:tvar}(c-d) demonstrates that PLSO is able to delineate the slow/theta components without any discontinuity.

\subsection{EEG data from the human brain under propofol anesthesia}
We apply PLSO to the EEG data from a subject anesthetized with propofol anesthetic drug, to assess whether PLSO can leverage regularization to recover smooth spectral dynamics, which is widely-observed during propofol-induced unconsciousness\footnote{The EEG recording is part of de-identified data collected from patients at Massachusetts General Hospital (MGH) as a part of a MGH
	Human Research Committee-approved protocol.}~\citep{purdon13}. The data last $T=2{,}300$ seconds, sampled at $f_s=250$ Hz. The drug infusion starts at $t=0$ and the subject loses consciousness around $t=260$ seconds. We use $J=6$ and assume a 4-second PS interval.

\begin{figure}[!ht]
	\centering
	\includegraphics[width=\linewidth]{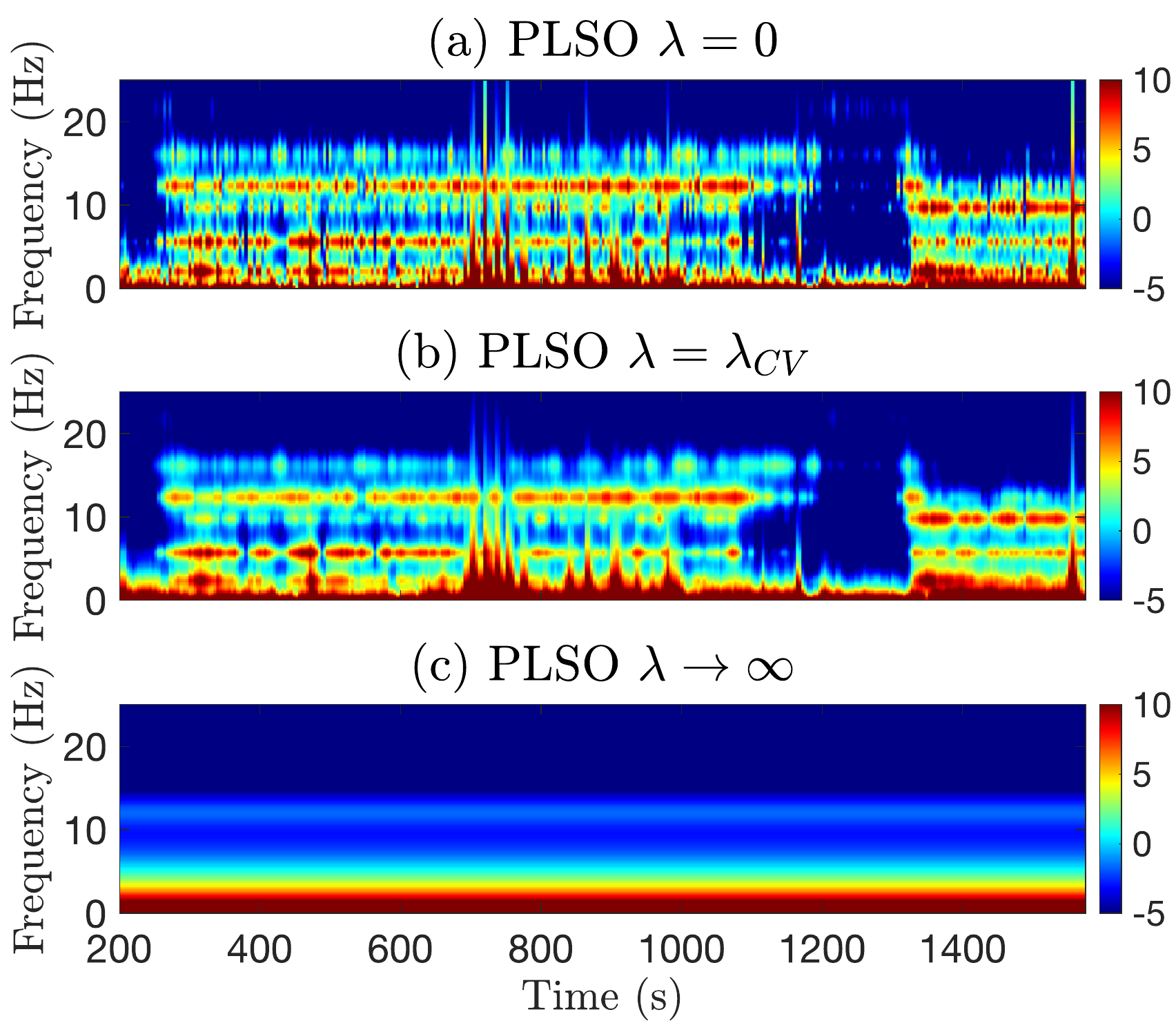}
	\caption{Spectrogram (in dB) under propofol anesthesia. PLSO with (a) $\lambda=0$ (b) $\lambda=\lambda_{\operatorname{CV}}$ (c) $\lambda\rightarrow\infty$.}
	\label{fig:anesthesia}
\end{figure}

\textbf{Smooth spectrogram} Fig.~\ref{fig:anesthesia}(a-b) shows a segment of the PLSO-estimated spectrogram with $\lambda=0$ and $\lambda=\lambda_{\operatorname{CV}}$. They identify strong slow ($0.1\sim 2$ Hz) and alpha oscillations ($8\sim 15$ Hz), both well-known signatures of propofol-induced unconsciousness. We also observe that the alpha band power diminishes between $1{,}200$ and $1{,}350$ seconds, suggesting that the subject regained consciousness before becoming unconscious again. 
%With $\lambda\rightarrow \infty$ (not shown), the time-varying dynamics would not be observed. 
 PLSO with $\lambda=0$ exhibits PSD fluctuation across windows, since $\sigset$ are estimated independently. The stationary PLSO ($\lambda\rightarrow \infty$) is restrictive and fails to capture spectral dynamics (Fig.~\ref{fig:anesthesia}(c)).  In contrast, PLSO with $\lambda_{\operatorname{CV}}$ exhibits smooth dynamics by pooling together estimates from the neighboring windows. The regularization also helps remove movement-related artifacts, shown as vertical lines in Fig.~\ref{fig:anesthesia}(a), around $700\sim 800$/$1{,}200$ seconds, and spurious power in $20\sim 25$ Hz band. In summary, PLSO with regularization enables smooth spectral dynamics estimation and spurious noise removal.

\section{Conclusion}\label{section:concolusion}
We presented the Piecewise Locally Stationary Oscillatory (PLSO) framework to model nonstationary time-series data with slowly time-varying spectra, as the superposition of piecewise stationary (PS) oscillatory components. PLSO strikes a balance between stochastic continuity of the data across PS intervals and stationarity within each interval. For inference, we introduce an algorithm that combines Kalman theory and nonconvex optimization algorithms. Applications to simulated/real data show that PLSO preserves time-domain continuity and captures time-varying spectra. Future directions include 1) the automatic identification of PS intervals and 2) the expansion to higher-order autoregressive models and diverse priors on the parameters that enforce continuity across intervals.

%\section{Back Matter}
%There are a some final, special sections that come at the back of the paper, in the following order:
%\begin{itemize}
%  \item Author Contributions
%  \item Acknowledgements
%  \item References
%\end{itemize}
%They all use an unnumbered \verb|\subsubsection|.
%
%For the first two special environments are provided.
%(These sections are automatically removed for the anonymous submission version of your paper.)
%The third is the ‘References’ section.
%(See below.)
%
%(This ‘Back Matter’ section itself should not be included in your paper.)

%\begin{contributions} % will be removed in pdf for initial submission,
%                      % so you can already fill it to test with the
%                      % ‘accepted’ class option
%     AHS, DB, and ENB conceived the idea. AHS wrote the paper and created the code. DB and ENB edited the paper.
%\end{contributions}

\begin{acknowledgements} % will be removed in pdf for initial submission,
                         % so you can already fill it to test with the
                         % ‘accepted’ class option
    We thank Abiy Tasissa for helpful discussion on the algorithm and the derivations. This research is supported by NIH grant P01 GM118629 and support from
    the JPB Foundation. AHS acknowledges Samsung scholarship.
\end{acknowledgements}

\bibliography{stochastic}

\clearpage
\section*{Appendix}
%References to the equations are made with respect to the equations in the supplementary materials. 
References to the sections in the section headings are made with respect to the sections in the main text. Below is the table of contents for the Appendix.

\begin{itemize}
	\item[A.] Continuous model interpretation of PLSO \textit{(Section 3)}
	\item[B.] PSD for complex AR(1) process
	\item[C.] Proof for Proposition 1 \textit{(Section 3.2.1)}
	\item[D.] Proof for Proposition 2 \textit{(Section 3.2.2)}
	\item[E.] Initialization for PLSO \textit{(Section 4)}
	\item[F.] Optimization for $\sigset$ via proximal gradient update \textit{(Section 4.1.1)}
	\item[G.] Lipschitz constant for $\nabla f(\sigset;\theta)$ \textit{(Section 4.1.1)}
	\item[H.] Inference with $p(\left\{\z_j\right\}_j\mid\sigset,\y,\theta)$ \textit{(Section 4.2)}
	\item[I.] Computational efficiency of PLSO vs. GP-PS 
	\item[J.] Simulation experiment \textit{(Section 5.1)}
	\item[K.] Details of TVAR model \textit{(Section 5.2)}
	\item[L.] Anesthesia EEG dataset \textit{(Section 5.3)}
	\end{itemize}

\section*{A. Continuous model interpretation of PLSO}
We can establish the equivalent continuous model of the PLSO in Eq.~\ref{eq:PSLO_stationary}, using stochastic different equation 
\begin{equation}\label{eq:continuous_PLSO}
\frac{d\tilde{z}_j(t)}{dt} = \underbrace{\left(\left(-\frac{1}{l_j}\right)\oplus \begin{pmatrix}
	0 & -\omega_j\\
	\omega_j & 0\\
	\end{pmatrix}\right)}_{\mathbf{F}}\tilde{z}_j(t) +\varepsilon(t),
\end{equation}
where $\tilde{z}_j(t): \mathbb{R}\rightarrow \mathbb{R}^2$, $\oplus$ denotes the Kronecker sum and $\varepsilon(t)\sim\N(0,\sigma_j^2\I_{2\times 2})$. Discretizing the solution of Eq.~\ref{eq:continuous_PLSO} at $\Delta$, such that $\zvec_{j,k}=\tilde{z}_j(k\Delta)$, yields Eq.~\ref{eq:PSLO_stationary}. Consequently, we obtain the following for $\Delta>0$
\begin{equation}\label{eq:matexp}
\begin{split}
&\exp\left(\mathbf{F}\Delta\right)=\exp(-\Delta/l_j)\mathbf{R}(\omega_j),\\
&\sigma_j^2\int_0^{\Delta}\exp\left(\mathbf{F}(\Delta-\tau)\right)\exp\left(\mathbf{F}(\Delta-\tau)\right)^{\text{T}}d\tau\\
&\quad\quad=\sigma_j^2\left(1-\exp\left(-2\Delta/l_j\right)\right)\I_{2\times 2}.\\
\end{split}
\end{equation}

This interpretation extends to the nonstationary PLSO. The corresponding continuous model for $\zvec_{j,mN+n}$ in Eq.~\ref{eq:tfgpss2} is the same as Eq.~\ref{eq:continuous_PLSO}, with different variance $\E[\varepsilon_j(t)\varepsilon_j^{\text{T}}(t)]=\sum_{m=1}^M\sigma_{j,m}^2\cdot\mathbf{1}\left(\left(\frac{m-1}{M}\right)T\leq t <\left(\frac{m}{M}\right)T\right)\I_{2\times 2}.$

\section*{B. PSD for complex AR(1) process}
We compute the steady-state covariance denoted as $\var_{\infty}^{j}$. Since we assume $\var_{1}^j=\sigma_{j}^2\I_{2\times 2}$, it is easy to show that $\var_{k}^j$ is a diagonal matrix from $\mathbf{R}(\omega_j)\mathbf{R}^{\text{T}}(\omega_j)=\I_{2\times 2}$. Denoting $\var_{\infty}^{j}=\alpha\I_{2\times 2}$, we use the discrete Lyapunov equation
\begin{equation}
\begin{split}
&\var_{\infty}^{j} \\
&= \exp(-2\Delta/l_j)\mathbf{R}(\omega_j)\var_{\infty}^{j}\mathbf{R}^{\text{T}}(\omega_j)\\
&\quad+\sigma_{j}^2\left(1-\exp\left(-2\Delta/l_j\right)\right)\I_{2\times 2}\\
&\Rightarrow \alpha = \exp(-2\Delta/l_j)\alpha+\sigma_{j}^2\left(1-\exp\left(-2\Delta/l_j\right)\right)\\
&\Rightarrow \var_{\infty}^{j} =\sigma_{j}^2\I_{2\times 2},\\
\end{split}
\end{equation}
which implies that under the assumption $\var_{1}^j=\sigma_{j}^2\I_{2\times 2}$, we are guaranteed $\var_{k}^j=\sigma_{j}^2\I_{2\times 2}$, $\forall k$. To compute the PSD of the stationary process $\z_j$, we now need to compute the autocovariance. Since only $\realz_{j,k}$ contributes to $\y_k$, we compute the autocovariance of $\E[ \realz_{j,k}\realz_{j,k+n}]$ as
\begin{equation}\label{eq:autocovariance}
\begin{split}
\E[ \realz_{j,k} \realz_{j,k+n}]&=\E[ \realz_{j,k}\cdot \Re(\rho_j^{n}\exp(i\omega_jn)\z_{j,k})]\\
&=\rho_j^{n}\E[ \realz_{j,k} \realz_{j,k}\cos \omega_jn]\\
&=\rho_j^{n}\cos \omega_jn\cdot \E[\{\realz_{j,k} \}^2]\\
&=\rho_j^{n}\sigma_j^2\cos w_jn,\\
\end{split}
\end{equation}
where $\Re(\cdot)$ denotes the operator that extracts the real part of the complex argument and we used the fact that $\E[\realz_{j,k}\imz_{j,k}]=0$. The spectra for the $j^{\text{th}}$ component,  $S_j(\omega)$ can be written as
\begin{equation}
\begin{split}
&S_j(\omega)=\sum_{n=-\infty}^{\infty}\E\left[ \realz_{j,k} \realz_{j,k+n}\right]\exp\left(-i\omega n\right)\\
&=\sum_{n=-\infty}^{\infty}\rho_j^{n}\sigma_j^2\cos w_jn\exp\left(-i\omega n\right) \\
&=\sigma_j^2\sum_{n=-\infty}^{\infty}\rho_j^{n}\left\{\exp(i\omega_jn)+\exp(-i\omega_jn) \right\}\exp\left(-i\omega n\right)\\
&=\sigma_j^2\sum_{n=-\infty}^{\infty}\rho_j^{n}\exp(-i(\omega\pm \omega_j)n).\\
\end{split}
\end{equation}
Unpacking the infinite sum for one of the terms,
\begin{equation}
\begin{split}
&\sum_{n=-\infty}^{\infty}\rho_j^n\exp(-i(\omega - \omega_j)n)\\
&=1+\sum_{n=1}^{\infty}\rho_j^{n}\exp(-i(\omega-\omega_j)n)+\rho_j^{n}\exp(i(\omega-\omega_j)n)\\
&=1+\frac{\rho_j\exp(-i(\omega- \omega_j))}{1-\rho_j\exp(-i(\omega- \omega_j))}+\frac{\rho_j\exp(i(\omega- \omega_j))}{1-\rho_j\exp(i(\omega- \omega_j))}\\
&=1+\frac{2\rho_j\cos (\omega-\omega_j)-2\rho_j^2}{\left(1-\rho_j\exp(-i(\omega - \omega_j))\right)\left(1-\rho_j\exp(i(\omega - \omega_j)) \right)}\\
&=\frac{1-\rho_j^2}{1+\rho_j^2-2\rho_j\cos(\omega-\omega_j)}.\\
\end{split}
\end{equation}
Using the relation $\rho_j=\exp(-\Delta/l_j)$ and unpacking the infinite sum for the other term, we have
\begin{equation*}
\begin{split}
&S_j(w)\\
&=\frac{\sigma_j^2(1-\exp(-2\Delta/l_j))}{1+\exp(-2\Delta/l_j)-2\exp(-\Delta/l_j)\cos(\omega-\omega_j)}\\
&\quad +\frac{\sigma_j^2(1-\exp(-2\Delta/l_j))}{1+\exp(-2\Delta/l_j)-2\exp(-\Delta/l_j)\cos(\omega+\omega_j)}.\\
\end{split}
\end{equation*}
Since Fourier transform is a linear operator, we can conclude that $\gamma(\omega)=\sigma_{\nu}^2+\sum_{j=1}^JS_j(\omega)$.

\section*{C. Proof for Proposition 1 \textit{(Section 3.2.1)}}
\textbf{Proposition 1.} For a given $m$, as $\Delta\rightarrow 0$, the samples on either side of the interval boundary, which are $\zvec_{j,(m+1)N}$ and $\zvec_{j,(m+1)N+1}$, converge to each other in mean square,
	\begin{equation*}
	\lim_{\Delta\rightarrow 0}\E[\Delta \zvec_{j, (m+1)N}\Delta \zvec_{j, (m+1)N}^{\text{T}}]=0,
	\end{equation*}
	where we use $\Delta \zvec_{j, (m+1)N}=\zvec_{j, (m+1)N+1}-\zvec_{j, (m+1)N}$.
\begin{proof} To analyze Eq.~\ref{eq:tfgpss2} in the limit of $\Delta\rightarrow 0$, we use the equivalent continuous model (Eq.~\ref{eq:matexp}). It suffices to show that $\lim_{\Delta\rightarrow 0}\exp(\mathbf{F}\Delta)=\I_{2\times 2}$ and $\lim_{\Delta\rightarrow 0}\E[\varepsilon_{j, (m+1)N+1}\varepsilon_{j, (m+1)N+1}^{\text{T}} ]=\mathbf{0}$. We have,
\begin{equation*}
\begin{split}
&\lim_{\Delta\rightarrow 0}\exp(\mathbf{F}\Delta)=\I_{2\times 2}+\lim_{\Delta\rightarrow 0 }\sum_{k=1}^{\infty}\frac{\Delta^k}{k!}\mathbf{F}^k=\I_{2\times 2}\\
&\lim_{\Delta\rightarrow 0}\E[\varepsilon_{j, (m+1)N+1}\varepsilon_{j, (m+1)N+1}^{\text{T}}  ]/\sigma_{j,m+1}^2\\
&=\lim_{\Delta\rightarrow 0}\int_0^{\Delta}\exp\left(\mathbf{F}(\Delta-\tau)\right)\exp\left(\mathbf{F}(\Delta-\tau)\right)^{\text{T}}d\tau=\mathbf{0}.\\
\end{split}
\end{equation*}
Since this implies $\lim_{\Delta\rightarrow 0}\E[\Delta \zvec_{j, (m+1)N}\Delta \zvec_{j, (m+1)N}^{\text{T}}]=0$, we have convergence in mean square.
\end{proof}

\section*{D. Proof for Proposition 2 \textit{(Section 3.2.2)}}
\textbf{Proposition 2.} Assume $l_j \ll N\Delta$, such that $\var_{m,N}^j=\var_{m,\infty}^{j}$. In Eq.~\ref{eq:tfgpss2}, the difference between $\var_{m,\infty}^{j}=\sigma_{j,m}^2\I_{2\times 2}$ and $\var_{m+1,\infty}^{j}=\sigma_{j,m+1}^2\I_{2\times 2}$ decays exponentially fast as a function of $n$,	for $1\leq n \leq N$,
\begin{equation*}\label{eq:steady_state_induction}
\begin{split}
\var_{m+1,n}^j&=\var_{m+1,\infty}^{j} + \exp\left(-2n\Delta/l_j\right)( \var_{m,\infty}^{j}-\var_{m+1,\infty}^{j}).\\
\end{split}
\end{equation*}

\begin{proof}
	We first obtain the steady-state covariance $\var_{m,\infty}^{j}$, similar to \textbf{Appendix B}. Since we assume $\var_{1,1}^j=\sigma_{j,1}^2\I_{2\times 2}$, we can show that $\forall m,n$, $\var_{m,n}^j$ is a diagonal matrix, noting that $\mathbf{R}(\omega_j)\mathbf{R}^{\text{T}}(\omega_j)=\I_{2\times 2}$. Denoting $\var_{m,\infty}^{j}=\alpha\I_{2\times 2}$, we now use the discrete Lyapunov equation
	\begin{equation}
	\begin{split}
	&\var_{m,\infty}^{j} = \exp(-2\Delta/l_j)\mathbf{R}(\omega_j)\var_{m,\infty}^{j}\mathbf{R}^{\text{T}}(\omega_j)\\
	&\quad+\sigma_{j,m}^2\left(1-\exp\left(-2\Delta/l_j\right)\right)\I_{2\times 2}\\
	&\Rightarrow \alpha = \exp(-2\Delta/l_j)\alpha+\sigma_{j,m}^2\left(1-\exp\left(-2\Delta/l_j\right)\right)\\
	&\Rightarrow \var_{m,\infty}^{j} =\sigma_{j,m}^2\I_{2\times 2}.\\
	\end{split}
	\end{equation}
	We now prove the proposition by induction. For fixed $j$ and $m$, and for $n=1$,
	\begin{equation*}
	\begin{split}
	&\var_{m+1,1}^j\\
	&=\exp(-2\Delta/l_j)\mathbf{R}(\omega_j)\var_{m,N}^{j}\mathbf{R}^{\text{T}}(\omega_j)\\
	&\quad+\sigma_{j,m+1}^2\left(1-\exp\left(-2\Delta/l_j\right)\right)\I_{2\times 2}\\
	&=\left\{\sigma_{j,m+1}^2+\exp\left(-2\Delta/l_j\right)\left(\sigma_{j,m}^2-\sigma_{j,m+1}^2\right)\right\}\I_{2\times 2}.\\
	\end{split}
	\end{equation*}
	Assuming the same holds for $n=n'-1$, we have for $n=n'$,
	\begin{equation*}
	\begin{split}
	&\var_{m+1,n'}^{j}\\
	&=\exp(-2\Delta/l_j)\mathbf{R}(\omega_j)\var_{m,n'-1}^{j}\mathbf{R}^{\text{T}}(\omega_j)\\
	&\quad+\sigma_{j,m+1}^2\left(1-\exp\left(-2\Delta/l_j\right)\right)\I_{2\times 2}\\
	&=\exp(-2\Delta/l_j)\sigma_{j,m+1}^2\I_{2\times 2}\\
	&\quad+\exp\left(-2n'\Delta/l_j\right)\left(\sigma_{j,m}^2-\sigma_{j,m+1}^2\right)\I_{2\times 2}\\
	&\quad+\sigma_{j,m+1}^2\left(1-\exp\left(-2\Delta/l_j\right)\right)\I_{2\times 2}\\
	&=\left\{\sigma_{j,m+1}^2+\exp\left(-2n'\Delta/l_j\right)\left(\sigma_{j,m}^2-\sigma_{j,m+1}^2\right)\right\}\I_{2\times 2}.\\
	\end{split}
	\end{equation*}
	By the principle of induction, Eq.~\ref{eq:steady_state_induction} holds for $1\leq n \leq N$.
\end{proof}

\section*{E. Initialization \& Estimation for PLSO \textit{(Section 4)}}
\subsection*{E.1 Initialization}
As noted in the main text, we use AIC to determine the optimal number of $J$. For a given number of components $J$, we first construct the spectrogram of the data using STFT and identify the frequency bands with prominent power, i.e., frequency bands whose average power exceeds pre-determined threshold. The center frequencies of these bands serve as the initial center frequencies $\{\omega_{j}^{\text{init}}\}_j$, which are either fixed throughout the algorithm or further refined through the estimation algorithm in the main text. If $J$ exceeds the number of identified frequency bands from the spectrogram, 1) we first place $\{\omega_j^{\text{init}}\}_j$ in the prominent frequency bands and 2) we then place the remaining components uniformly spread out in $[0,\omega_c]$, where $\omega_c$ is a cutoff frequency to be further determined in the next section. As for $\{l_j^{\text{init}}\}_j$, we set it to be a certain fraction of the corresponding $\{\omega_{j}^{\text{init}}\}_j$. We then fit $\sigset$ and $\theta$ with $\lambda=0$, through the procedure explained in Stage 1. We finally use these estimates as initial values for other values of $\lambda$. 

\subsection*{E.2 Estimation for $\sigma_{\nu}^2$}
There are two possible ways to estimate the observation noise variance $\sigma_{\nu}^2$. The first approach is to perform maximum likelihood estimation of $f(\sigset;\theta)$ with respect to $\sigma_{\nu}^2$. The second approach, which we found to work \textit{better} in practice and use throughout the manuscript, is to directly estimate it from the Fourier transform of the data. Given a cutoff frequency $\omega_c$, informed by domain knowledge, we take the average power of the Fourier transform of $\y$ in $[\omega_c, f_s/2]$. For instance, it is widely known that the spectral content below 40 Hz in anesthesia EEG dataset is physiologically relevant and we use $\omega_c\simeq 40$ Hz.

\section*{F. Optimization for $\sigset$ via proximal gradient update \textit{(Section 4.1.1)}}
We discuss the algorithm to obtain a local optimal solution of $\sigset$ to the MAP optimization problem in Eq.~\ref{eq:penalized_whittle}. We define $\pow_{j,m}=\log\sigma_{j,m}^2$ and $\Pow=[\pow_{1,1},\ldots,\pow_{1,M},\ldots,\pow_{J,M}]\in\mathbb{R}^{JM}$ for notational decluttering, to be used in this section.

We rewrite Eq.~\ref{eq:penalized_whittle} as
\begin{equation}
\begin{split}
&\min_{\Pow} \underbrace{-\log p(\Pow|\y,\theta)}_{h(\Pow;\theta)}\\
&=\min_{\Pow} \underbrace{\frac{1}{2}\sum_{m=1}^M\sum_{n=1}^N\left\{\log\gamma^{(m)}(\omega_n)+\frac{I^{(m)}(\omega_n)}{\gamma^{(m)}(\omega_n)}\right\}}_{-f(\Pow;\theta)}\\
&\quad\quad\quad+\underbrace{\frac{\lambda}{2}\sum_{j=1}^J\sum_{m=1}^M \left(\pow_{j,m} - \pow_{j,m-1} \right)^2}_{-g(\Pow;\theta)}\\
&=\min_{\Pow}-f(\Pow;\theta)-g(\Pow;\theta).\\
\end{split}
\end{equation}
The algorithm is described in Algorithm~\ref{alg:apg}. It follows the steps outlined in the inexact accelerated proximal gradient algorithm~\citep{Li15}. For faster convergence, we use larger step sizes with the Barzilai-Borwein (BB) step size initialization rule~\citep{barzilai88}. For rest of this section, we drop dependence on $\theta$. The main novelty of our algorithm is the proximal gradient update
\begin{equation}\label{eq:proximal2}
\begin{split}
&\mathbf{u}^{(l+1)}\\
&=\operatorname{prox}_{-\alpha^{(l)}_{\mathbf{w}} g}( \mathbf{w}^{(l)} + \alpha_{\mathbf{w}}^{(l)}\nabla f(\mathbf{w}^{(l)}) )\\
&=\argmin_{\Pow}\frac{1}{2\alpha_{\mathbf{w}}^{(l)}}\lVert\Pow-(\mathbf{w}^{(l)} + \alpha_{\mathbf{w}}^{(l)}\nabla f(\mathbf{w}^{(l)}))\rVert^2-g(\Pow)\\
&=\argmin_{\Pow}\sum_{j=1}^{J} \sum_{m=1}^M \frac{\left( \left( \mathbf{w}_{j,m}^{(l)} + \alpha_{\mathbf{w}}^{(l)}\cdot\frac{\partial f(\mathbf{w}^{(l)})}{\partial \mathbf{w}_{j,m}} \right)- \pow_{j,m}\right)^2}{2\alpha_{\mathbf{w}}^{(l)}}\\
&\quad\quad\quad\quad\quad\quad\quad+\frac{\lambda}{2}(\pow_{j,m}-\pow_{j,m-1})^2,\\
\end{split}
\end{equation}
where the same holds for $\mathbf{x}^{(l+1)}=\operatorname{prox}_{-\alpha_{\Pow}^{(l)}g}( \Pow^{(l)}+\alpha_{\Pow}^{(l)}\nabla f(\Pow^{(l)}))$. The auxiliary variables $\mathbf{w},\mathbf{u},\mathbf{x}\in\mathbb{R}^{JM}$ ensure convergence of $\Pow$. We use $\mathbf{w}^{(l)}_{j,m}$ to denote $\left((m-1)J+j\right)^{\text{th}}$ entry of $\mathbf{w}^{(l)}$. As mentioned in the main text, this can be solved in a computationally efficient manner by using Kalman filter/smoother. 

\begin{algorithm}
	\SetAlgoLined
	\SetKwRepeat{Repeat}{repeat}{until}
	\KwResult{$\widehat{\Pow}$}
	\textbf{Initialize} $\Pow^{(0)}=\Pow^{(1)}=\mathbf{u}^{(1)}$, $\beta^{(0)}=0,\,\beta^{(1)}=1,\,\delta>0,\,\rho<1$ \\
	\For{$l\leftarrow 1$ \KwTo $L$}{
		$\mathbf{w}^{(l)}=\Pow^{(l)}+\frac{\beta^{(l-1)}}{\beta^{(l)}}(\mathbf{u}^{(l)}-\Pow^{(l)})+\frac{\beta^{(l-1)}-1}{\beta^{(l)}}(\Pow^{(l)}-\Pow^{(l-1)})$\\
		\textbf{(BB step size initialization rule)}\\
		$\mathbf{s}^{(l)}= \mathbf{u}^{(l)}-\mathbf{w}^{(l-1)}$, $\mathbf{r}^{(l)}=-\nabla f(\mathbf{u}^{(l)})+\nabla f(\mathbf{w}^{(l-1)})$\\
		$\alpha_{\mathbf{w}}^{(l)}= ((\mathbf{s}^{(l)})^{\text{T}}\mathbf{s}^{(l)})/((\mathbf{s}^{(l)})^{\text{T}}\mathbf{r}^{(l)})$\\
		$\mathbf{s}^{(l)}= \mathbf{x}^{(l)}-\Pow^{(l-1)}$, $\mathbf{r}^{(l)}=-\nabla f(\mathbf{x}^{(l)})+\nabla f(\Pow^{(l-1)})$	\\
		$\alpha_{\Pow}^{(l)}= ((\mathbf{s}^{(l)})^{\text{T}}\mathbf{s}^{(l)})/((\mathbf{s}^{(l)})^{\text{T}}\mathbf{r}^{(l)})$\\
		\textbf{(Proximal update step)}\\
		\Repeat{$h(\mathbf{u}^{(l+1)})\leq h(\mathbf{w}^{(l)})-\delta\lVert\mathbf{u}^{(l+1)}- \mathbf{w}^{(l)} \rVert^2$}{
			\quad$\mathbf{u}^{(l+1)}=\operatorname{prox}_{-\alpha_{\mathbf{w}}^{(l)}g}( \mathbf{w}^{(l)} + \alpha_{\mathbf{w}}^{(l)}\nabla f(\mathbf{w}^{(l)}))$\\
			\quad$\alpha_{\mathbf{w}}^{(l)}=\rho\cdot\alpha_{\mathbf{w}}^{(l)}$\\
		}
		\Repeat{$h(\mathbf{x}^{(l+1)})\leq h(\Pow^{(l)})-\delta\lVert\mathbf{x}^{(l+1)}- \Pow^{(l)} \rVert^2$}{
			\quad$\mathbf{x}^{(l+1)}=\operatorname{prox}_{-\alpha_{\Pow}^{(l)}g}( \Pow^{(l)}+\alpha_{\Pow}^{(l)}\nabla f(\Pow^{(l)}))$\\
			\quad $\alpha_{\Pow}^{(l)}=\rho\cdot\alpha_{\Pow}^{(l)}$\\
		}
		$\beta^{(l+1)}=\frac{1 + \sqrt{4\left(\beta^{(l)} \right)^2+1}}{2}$\\
		$\Pow^{(l+1)}=\begin{cases}
		\mathbf{u}^{(l+1)}&\text{if }h(\textbf{u}^{(l+1)})\leq h(\textbf{x}^{(l+1)})\\
		\mathbf{x}^{(l+1)}& \text{otherwise }\\
		\end{cases}$
	}
	$\widehat{\Pow}=\Pow^{(L)}$
	\caption{Inference for $\Pow$ via inexact APG}
	\label{alg:apg}
\end{algorithm}

\newpage

\section*{G. Lipschitz constant for $\nabla f(\sigset;\theta)$ \textit{(Section 4.1.1)}}
In this section, we prove that under some assumptions, we can show that the log-likelihood $f(\sigset;\theta)$ has Lipschitz continuous gradient with the Lipschitz constant $C$. As in the previous section, we use $\pow_{j,m}=\log\sigma_{j,m}^2$ and $\Pow=[\pow_{1,1},\ldots,\pow_{1,M},\ldots,\pow_{J,M}]\in\mathbb{R}^{JM}$. 

Let us start by restating the definition of Lipschitz continuous gradient.

\textbf{Definition} A continuously differentiable function $f:\mathcal{S}\rightarrow \mathbb{R}$ is Lipschitz continuous gradient if
\begin{equation}
\lVert\nabla f(\mathbf{x})-\nabla f(\mathbf{y})\rVert_2 \leq C\lVert \textbf{x} - \textbf{y}\rVert_2\quad\text{for every } \mathbf{x},\mathbf{y}\in\mathcal{S},
\end{equation}
where $\mathcal{S}$ is a compact subset of $\mathbb{R}^{JM}$ and $C>0$ is the Lipschitz constant. 

Our goal is to find the constant $C$ for the Whittle likelihood $f(\Pow;\theta)$
\begin{equation}
\begin{split}
f(\Pow;\theta)&=-\frac{1}{2}\sum_{m=1}^M\sum_{n=1}^N\left\{\log\gamma_{m,n}+\frac{I_{m,n}}{\gamma_{m,n}}\right\}\\
&=-\frac{1}{2}\underbrace{\sum_{m=1}^M\sum_{n=1}^N\log\left(\sigma_{\nu}^2+\sum_{j=1}^J \exp\left(\pow_{j,m}\right)\alpha_{j,n} \right)}_{f_{1}(\Pow;\theta)}\\
&\quad-\frac{1}{2}\underbrace{\sum_{m=1}^M\sum_{n=1}^N\frac{I_{m,n}}{\left(\sigma_{\nu}^2+\sum_{j=1}^J \exp\left(\pow_{j,m}\right)\alpha_{j,n} \right)}}_{f_{2}(\Pow;\theta)},\\
\end{split}
\end{equation}
where we use $\gamma_{m,n}=\gamma^{(m)}(\omega_n)$ and $I_{m,n}=I^{(m)}(\omega_n)$ for notational simplicity and
\begin{equation*}
\alpha_{j,n}=\frac{\left(1 - \exp\left(-2\Delta/l_j\right)\right)}{1+\exp\left(-2\Delta/l_j\right)-2\exp\left(-\Delta/l_j\right)\cos(\omega_n-\omega_j)}.
\end{equation*}

We make the following assumptions
\begin{enumerate}
	\item $I_{m,n}$ is bounded, i.e., $I_{m,n}\leq C_I$. With the real-world signal, we can reasonably assume that $I_{m,n}$ or energy of the signal is bounded.
	\item $\forall j,m$, $\pow_{j,m}$ is bounded, i.e., $|\pow_{j,m}|\leq \log C_{\pow}$ for some $C_{\pow}>1$. This implies $1/C_{\pow}\leq \exp(\pow_{j,m})\leq C_{\pow}$. 
\end{enumerate}

In addition, we have the following facts
\begin{enumerate}
	\item $I_{m,n}, \alpha_{j,n}$, and $\gamma_{m,n}$ are nonnegative.
	\item $I_{m,n}$ and $\gamma_{m,n}$ are bounded. This follows from the aforementioned assumptions.
	\item For given $\{l_j\}_j$, we have bounded $\alpha_{j,n}$. To see this, note that the maximum of $\alpha_{j,n}$ is acheived at $\omega_n=\omega_j$,
	\begin{equation}
	\begin{split}
	\max \alpha_{j,n}&=\frac{\left(1 - \exp\left(-2\Delta/l_j\right)\right)}{1+\exp\left(-2\Delta/l_j\right)-2\exp\left(-\Delta/l_j\right)}\\
	&=\frac{\left(1+\exp(-\Delta/l_j) \right)}{\left(1-\exp(-\Delta/l_j) \right)}.\\
	\end{split}
	\end{equation}
	Therefore, denoting $l_{\operatorname{max}}=\max_j \{l_j\}_j$, 
	\begin{equation}
	\max \alpha_{j,n}\leq \frac{\left(1+\exp(-\Delta/l_{\operatorname{max}}) \right)}{\left(1-\exp(-\Delta/l_{\operatorname{max}}) \right)}=C_{\alpha}.
	\end{equation}
\end{enumerate}
Finally, we define $\mathcal{S}=[-\log C_{\pow}, \log C_{\pow}]\subset \mathbb{R}^{JM}$.

We want to compute the Lipschitz constant for $\nabla f_1(\Pow)$ and $\nabla f_2(\Pow)$ for $\Pow,\overline{\Pow}\in \mathcal{S}$, i.e.,
\begin{equation}
\begin{cases}
\lVert\nabla f_1(\Pow)-\nabla f_1(\overline{\Pow})\rVert_2\leq C_1\lVert\Pow-\overline{\Pow}\rVert_2\\
\lVert \nabla f_2(\Pow)-\nabla f_2(\overline{\Pow})\rVert_2\leq C_2\lVert\Pow-\overline{\Pow}\rVert_2,\\
\end{cases}
\end{equation}
where we dropped dependence on $\theta$ for notational simplicity. Consequently, the triangle inequality yields
\begin{equation}
\begin{split}
&\lVert \nabla f(\Pow)-\nabla f(\overline{\Pow})\rVert_2\\
&\leq \lVert\nabla f_1(\Pow)-\nabla f_1(\overline{\Pow})\rVert_2+\lVert\nabla f_2(\Pow)-\nabla f_2(\overline{\Pow})\rVert_2\\
&\leq (C_1+C_2)\lVert\Pow-\overline{\Pow}\rVert_2.
\end{split}
\end{equation}

\subsection*{Lipschitz constant $C_1$ for $\nabla f_1(\Pow)$}
Let us examine $\first$ first. The derivative with respect to $\pow_{j,m}$ is given as
\begin{equation}
\begin{split}
\frac{\partial \first}{\partial \pow_{j,m}}&=\sum_{n=1}^N \alpha_{j,n}\cdot\underbrace{\frac{\exp\left(\pow_{j,m}\right)}{ \sigma_{\nu}^2+\sum_{j'=1}^J \exp\left(\pow_{j',m}\right)\alpha_{j',n} }}_{\widetilde{f}(\pow_{j,m})}\\
&=\sum_{n=1}^N \alpha_{j,n}\widetilde{f}(\pow_{j,m}).
\end{split}
\end{equation}
We now have
\begin{equation}
\left| \frac{\partial f_1(\Pow)}{\partial \pow_{j,m}} - \frac{\partial f_1(\overline{\Pow})}{\partial \pow_{j,m}}\right|=\sum_{n=1}^N|\alpha_{j,n}|\cdot\left| \widetilde{f}\left(\pow_{j,m}\right)- \widetilde{f}(\overline{\pow}_{j,m})\right|.
\end{equation}
Without loss of generality, we assume $\pow_{j,m}\geq \overline{\pow}_{j,m}$. We now apply the mean value theorem (MVT) to $\widetilde{f}(\pow_{j,m})$
\begin{equation}\label{eq:MVT}
\widetilde{f}(\pow_{j,m}) - \widetilde{f}(\overline{\pow}_{j,m})=\widetilde{f}'(\pow'_{j,m})(\pow_{j,m}-\overline{\pow}_{j,m}),
\end{equation} 
where $\pow'_{j,m}\in[\overline{\pow}_{j,m},\pow_{j,m}]$.We can compute and bound $\widetilde{f}'(\pow_{j,m}')=d \widetilde{f}(\pow_{j,m}')/d\pow_{j,m}'$ as follows
\begin{equation}\label{eq:derivative}
\begin{split}
\widetilde{f}'(\pow_{j,m}')&=\frac{\exp(\pow_{j,m}')}{ \sigma_{\nu}^2+\sum_{j'=1}^J \exp(\pow_{j',m}')\alpha_{j',n} }\\
&\quad - \frac{\alpha_{j,n}\exp^2(\pow_{j,m}')}{( \sigma_{\nu}^2+\sum_{j'=1}^J \exp(\pow_{j',m}')\alpha_{j',n})^2}\\
&=\frac{\exp(\pow_{j,m}')}{ \sigma_{\nu}^2+\sum_{j'=1}^J \exp(\pow_{j',m}')\alpha_{j',n} }\\
&\quad\times\underbrace{\left(1-\frac{\alpha_{j,n}\exp(\pow_{j,m}')}{ \sigma_{\nu}^2+\sum_{j'=1}^J \exp(\pow_{j',m}')\alpha_{j',n} }\right)}_{\leq 1}\\
&\leq \frac{\exp(\pow'_{j,m})}{\sigma_{\nu}^2+\sum_{j'=1}^J \exp(\pow'_{j',m})\alpha_{j',n}}\\
&\leq\frac{C_{\pow}}{\sigma_{\nu}^2}.\\
\end{split}
\end{equation}
Combining Eq.~\ref{eq:MVT} and \ref{eq:derivative}, we have
\begin{equation}
\begin{split}
&\sum_{n=1}^N|\alpha_{j,n}|\cdot\left| \widetilde{f}\left(\pow_{j,m}\right)- \widetilde{f}(\overline{\pow}_{j,m})\right|\\
&=\sum_{n=1}^N|\alpha_{j,n}|\cdot\left|\widetilde{f}'(\pow'_{j,m})(\pow_{j,m}-\overline{\pow}_{j,m})\right|\\
&\leq \frac{NC_{\alpha}C_{\pow}}{\sigma_{\nu}^2}\left| \pow_{j,m}-\overline{\pow}_{j,m}\right|.\\
\end{split}
\end{equation}

We thus have,
\begin{equation}
\begin{split}
\lVert\nabla f_1(\Pow)-\nabla f_1(\overline{\Pow})\rVert_2^2&=\sum_{j=1}^J\sum_{m=1}^M\left( \frac{\partial f_1(\Pow)}{\partial \pow_{j,m}} - \frac{\partial f_1(\overline{\Pow})}{\partial \pow_{j,m}}\right)^2\\
&\leq \left(\frac{JMNC_{\alpha}C_{\pow}}{\sigma_{\nu}^2}\right)^2\lVert\Pow-\overline{\Pow}\rVert^2_2.
\end{split}
\end{equation}

\subsection*{Lipschitz constant $C_2$ for $\nabla f_2(\Pow)$}
Computing $C_2$ proceeds in a similar manner to computing $C_1$. The derivative with respect to $\pow_{j,m}$ is given as 
\begin{equation}
\begin{split}
&\frac{\partial \second}{\partial \pow_{j,m}}\\
&= -\sum_{n=1}^N I_{m,n}\alpha_{j,n}\cdot\underbrace{\frac{\exp(\pow_{j,m})}{(\sigma_{\nu}^2+\sum_{j'=1}^J \exp(\pow_{j',m})\alpha_{j',n})^2}}_{\widetilde{f}(\pow_{j,m})}.\\
\end{split}
\end{equation}
We now have
\begin{equation}
\begin{split}
&\left| \frac{\partial f_2(\Pow)}{\partial \pow_{j,m}} - \frac{\partial f_2(\overline{\Pow})}{\partial \pow_{j,m}}\right|\\
&=\sum_{n=1}^N\left|I_{m,n}\alpha_{j,n}\right|\cdot\left| -\widetilde{f}(\pow_{j,m})+ \widetilde{f}(\overline{\pow}_{j,m})\right|.\\
%&\leq\sum_{n=1}^N  \frac{I_{m,n}\alpha_{j,n}}{\left(\sigma_{\nu}^2+\min\{\sum_{j'=1}^J \exp\left(\pow_{j',m}\right)\alpha_{j',n}, \sum_{j'=1}^J \exp\left(\overline{\pow}_{j',m}\right)\alpha_{j',n} \}\right)^2}\cdot\left| \exp(\pow_{j,m}) -\exp\left(\overline{\pow}_{j,m}\right)\right|\\
%&\leq \frac{NC_{\alpha}C_I}{\sigma_{\nu}^4}\left| \exp(\pow_{j, m}) -\exp\left(\overline{\pow}_{j,m}\right)\right|.\\
\end{split}
\end{equation}
Without loss of generality, assume $\pow_{j,m}\geq \overline{\pow}_{j,m}$. To apply MVT, we need to compute and bound $\widetilde{f}'(\pow'_{j,m})$
\begin{equation}\label{eq:f2_derivative}
\begin{split}
\widetilde{f}'(\pow_{j,m}')&=\frac{\exp(\pow_{j,m}')}{( \sigma_{\nu}^2+\sum_{j'=1}^J \exp(\pow_{j',m}')\alpha_{j',n})^2 }\\
&\quad - \frac{2\alpha_{j,n}\exp^2(\pow_{j,m}')}{( \sigma_{\nu}^2+\sum_{j'=1}^J \exp(\pow_{j',m}')\alpha_{j',n})^3}\\
&=\frac{\exp(\pow_{j,m}')}{( \sigma_{\nu}^2+\sum_{j'=1}^J \exp(\pow_{j',m}')\alpha_{j',n} )^2}\\
&\quad\times\underbrace{\left(1-\frac{2\alpha_{j,n}\exp(\pow_{j,m}')}{ \sigma_{\nu}^2+\sum_{j'=1}^J \exp(\pow_{j',m}')\alpha_{j',n} }\right)}_{\leq 1}\\
&\leq \frac{\exp(\pow'_{j,m})}{(\sigma_{\nu}^2+\sum_{j'=1}^J \exp(\pow'_{j',m})\alpha_{j',n})^2}\\
&\leq\frac{C_{\pow}}{\sigma_{\nu}^4}.\\
\end{split}
\end{equation}
Applying MVT,
\begin{equation}
\begin{split}
&\sum_{n=1}^N|I_{m,n}\alpha_{j,n}|\cdot\left| -\widetilde{f}\left(\pow_{j,m}\right)+ \widetilde{f}\left(\overline{\pow}_{j,m}\right)\right|\\
&=\sum_{n=1}^N|I_{m,n}\alpha_{j,n}|\cdot\left|\widetilde{f}'(\pow'_{j,m})(\pow_{j,m}-\overline{\pow}_{j,m})\right|\\
&\leq \frac{NC_IC_{\alpha}C_{\pow}}{\sigma_{\nu}^4}\left| \pow_{j,m}-\overline{\pow}_{j,m}\right|.\\
\end{split}
\end{equation}

%We now apply the MVT and obtain the following
%\begin{equation}
%\lVert\nabla f_2(\Pow)-\nabla f_2(\overline{\Pow})\rVert_2^2=\sum_{j=1}^J\sum_{m=1}^M\left( \frac{\partial f_2(\Pow)}{\partial \pow_{j,m}} - \frac{\partial f_2(\overline{\Pow})}{\partial \pow_{j,m}}\right)^2\leq \left(\frac{JMNC_{\alpha}C_{\pow}C_I}{\sigma_{\nu}^4}\right)^2\lVert\Pow-\overline{\Pow}\rVert_2^2.
%\end{equation}

We thus have,
\begin{equation}
\begin{split}
&\lVert\nabla f_2(\Pow)-\nabla f_2(\overline{\Pow})\rVert_2^2\\
&=\sum_{j=1}^J\sum_{m=1}^M\left( \frac{\partial f_2(\Pow)}{\partial \pow_{j,m}} - \frac{\partial f_2(\overline{\Pow})}{\partial \pow_{j,m}}\right)^2\\
&\leq \left(\frac{JMNC_{\alpha}C_{\pow}C_I}{\sigma_{\nu}^4}\right)^2\lVert\Pow-\overline{\Pow}\rVert^2_2.
\end{split}
\end{equation}

Collecting the Lipschitz constants $C_1$ and $C_2$, we finally have
\begin{equation}
\lVert \nabla f(\Pow)-\nabla f(\overline{\Pow})\rVert_2\leq \underbrace{\frac{JMNC_{\alpha}C_{\pow}}{\sigma_{\nu}^2}\left(1+\frac{C_I}{\sigma_{\nu}^2}\right)}_{C}\lVert\Pow-\overline{\Pow}\rVert_2.
\end{equation}

\section*{H. Inference with $p\left(\left\{\z_j\right\}_j\mid\sigset,\y,\theta\right)$ \textit{(Section 4.2)}}

We present the details for performing inference with $p(\left\{\z_j\right\}_j\mid \{\widehat{\sigma}_{j,m}^2\}_{j,m},\y,\hat{\theta})$, given the estimates $\{\widehat{\sigma}_{j,m}^2\}_{j,m}$ and $\hat{\theta}$ from window-level inference. First, we present the Kalman filter/smoother algorithm to compute the mean posterior trajectory, and the credible interval. Next, we present the forward filtering backward sampling (FFBS) algorithm~\cite{Carter94, Lindsten13} to generate Monte Carlo (MC) sample trajectories.

First, we define additional notations.

1) $\zvec_{j,k|k'}=\E[\zvec_{j,k} \mid \{\widehat{\sigma}_{j,m}^2\}_{j,m}, \y_{1:k'}, \hat{\theta} ]\in\mathbb{R}^2$ 

Posterior mean of $\zvec_{j,k}$. We are primarily concerned with the following three types: 1) $\zvec_{j,k|k-1}$, the one-step prediction estimate, 2) $\zvec_{j,k|k}$, the Kalman filter estimate, and 3) $\zvec_{j,k|MN}$, the Kalman smoother estimate.
	
2) $\zvec_{k|k'}=[(\zvec_{1,k|k'})^{\text{T}},\ldots, (\zvec_{J,k|k'})^{\text{T}}]^{\text{T}} \in\mathbb{R}^{2J}$ 

A collection of $\{\zvec_{j,k|k'} \}_j$ in a single vector.	

3) $\var_{j,k|k'}=\E[\left(\zvec_{j,k} - \zvec_{j,k|k'}\right)\left(\zvec_{j,k} - \zvec_{j,k|k'}\right)^{\text{T}}\mid \{\widehat{\sigma}_{j,m}^2\}_{j,m}, \y_{1:k'}, \hat{\theta}  ]\in\mathbb{R}^{2\times 2}$

Posterior covariance of $\zvec_{j,k}$. Just as in $\zvec_{j,k|k'}$, we are interested in three types, i.e., $\var_{j,k|k-1}$, $\var_{j, k|k}$, and $\var_{j,k|MN}$.
	
4) $\var_{k|k'}=\operatorname{blkdiag}\left( \var_{j,k|k'}\right)\in\mathbb{R}^{2J\times 2J}$ 

A block diagonal matrix of $J$ posterior covariance matrices.
	
5) $\mathbf{A}=\operatorname{blkdiag}\left(\exp\left(-\Delta/l_j\right)\mathbf{R}(\omega_j)\right)\in\mathbb{R}^{2J\times 2J}$

A block digonal transition matrix.
	
6) $\mathbf{H}=(1,0,\ldots,1,0)$ The observation gain. 

\subsection*{Kalman filter/smoother}
The Kalman filter equations are given as
\begin{equation}
\begin{split}
&\zvec_{j,mN+n|mN+(n-1)}\\ 
&= \exp\left(-\Delta/l_j\right)\mathbf{R}(\omega_j) \zvec_{j,mN+(n-1)|mN+(n-1)}\\ 
&\var_{j,mN+n|mN+(n-1)}\\
&=\exp\left(-2\Delta/l_j\right)\mathbf{R}(\omega_j)\var_{j,mN+(n-1)|mN+(n-1)}\mathbf{R}^{\text{T}}(\omega_j)\\
&\quad+\sigma_{j,m}^2\left(1-\exp\left(-2\Delta/l_j\right)\right)\\
&\K_{mN+n}\\
&=\var_{mN+n|mN+(n-1)}\mathbf{H}^{\text{T}}\left(\mathbf{H}\var_{mN+n|mN+(n-1)}\mathbf{H}^{\text{T}}+\sigma_{\nu}^2\right)^{-1}\\
&\zvec_{mN+n|mN+n}\\
&=\zvec_{mN+n|mN+(n-1)}\\ 
&\quad + \K_{mN+n}\left(\y_{mN+n}-\mathbf{H}\zvec_{mN+n|mN+(n-1)}\right)\\
&\var_{mN+n|mN+n}=\left(\I_{2J\times 2J}-\K_{mN+n}\mathbf{H}\right)\var_{mN+n|mN+(n-1)}.\\
\end{split}
\end{equation}
Subsequently, the Kalman smoother equations are given as
\begin{equation}
\begin{split}
&\mathbf{C}_{mN+n}\\
&=\var_{mN+n|mN+n}\mathbf{A}^{\text{T}}\var^{-1}_{mN+(n+1)|mN+n}\in\mathbb{R}^{2J\times 2J}\\
&\zvec_{mN+n|MN}\\
&= \zvec_{mN+n|mN+n}\\
&\quad+\mathbf{C}_{mN+n}\left(\zvec_{mN+(n+1)|MN} -\zvec_{mN+(n+1)|mN+n} \right)\\
&\var_{mN+n|MN}\\
&=\var_{mN+n|mN+n}\\
&\quad+\mathbf{C}_{mN+n}\var_{mN+(n+1)|MN}\mathbf{C}_{mN+n}^{\text{T}}.\\
&\quad-\mathbf{C}_{mN+n}\var_{mN+(n+1)|mN+n}\mathbf{C}_{mN+n}^{\text{T}}\\
\end{split}
\end{equation}
To obtain the mean reconstructed trajectory for the $j^{\text{th}}$ oscillatory component, $\{\widehat{\y}_{j,k}\}_k$, we take the real part of the $j^{\text{th}}$ component of the smoothed mean, $\widehat{\y}_{j,k}=\mathbf{e}_{2j-1}^{\text{T}}\cdot\zvec_{k|MN}$,  where $\mathbf{e}_{2j-1}\in\mathbb{R}^{2J}$ is a unit vector with the only non-zero value, equal to 1, at the entry $2j-1$.

The 95\% credible interval for $\widehat{\y}_{j,k}$, denoted as $\operatorname{CI}^{\text{lower}}_{j,k}$/$\operatorname{CI}^{\text{upper}}_{j,k}$ for the upper/lower end, respectively, is given as 
\begin{equation}
\begin{split}
\operatorname{CI}^{\text{upper}}_{j,k}&=\mathbf{e}_{2j-1}^{\text{T}}\cdot\zvec_{k|MN}+1.96\cdot \sqrt{\mathbf{e}_{2j-1}^{\text{T}}\var_{k|MN}\mathbf{e}_{2j-1}}\\
\operatorname{CI}^{\text{lower}}_{j,k}&=\mathbf{e}_{2j-1}^{\text{T}}\cdot\zvec_{k|MN}-1.96\cdot \sqrt{\mathbf{e}_{2j-1}^{\text{T}}\var_{k|MN}\mathbf{e}_{2j-1}}.\\
\end{split}
\end{equation}

\subsection*{FFBS algorithm for $p(\left\{\z_j\right\}_j\mid\sigset,\y,\theta)$}
To generate the MC trajectory samples from the posterior distribution $p(\left\{\z_j\right\}_j\mid \{\widehat{\sigma}_{j,m}^2\}_{j,m}, \y, \hat{\theta} )$, we use the FFBS algorithm. The steps are summarized in Algorithm~\ref{alg:ffbs}, which uses the Kalman estimates derived in the previous section. We denote $s=1,\ldots,S$ as the MC sample index.

\begin{algorithm}
	\SetAlgoLined
	\KwResult{$\left\{ \widetilde{\z}_k^{(s)} \right\}_{k,s}^{MN,S}$}
	\For{$s \leftarrow 1$ \KwTo $S$}{
		Sample $\widetilde{\z}^{(s)}_{MN}$ from $\mathcal{N}\left(\zvec_{MN|MN}, \var_{MN|MN}\right)$.\\
		\For{$k\leftarrow MN-1$ \KwTo $1$}{
			Sample $\widetilde{\z}^{(s)}_{k}$ from $\mathcal{N}(\widetilde{\mu}_k,\widetilde{\var}_k)$, where
			\begin{equation*}
			\begin{split}
			\widetilde{\mu}_k&=\zvec_{k|k}+\var_{k|k}\mathbf{A}^{\text{T}}\var_{k+1|k}^{-1}( \widetilde{\z}^{(s)}_{k+1}-\mathbf{A}\zvec_{k|k})\\
			\widetilde{\var}_k&=\var_{k|k}-\var_{k|k}\mathbf{A}^{\text{T}}\var_{k+1|k}^{-1}\mathbf{A}\var_{k|k}.\\
			\end{split}
			\end{equation*}
		}
	}
	\caption{FFBS algorithm}
	\label{alg:ffbs}
\end{algorithm}

\newpage

\section*{I. Computational efficiency of PLSO vs. GP-PS}
We show the runtime of PLSO and piecewise stationary GP (GP-PS) for inference of the mean trajectory of the hippocampus data ($fs = 1,250$ Hz, $J = 5$, 2-second window) for varying data lengths ($50$, $100$, $200$ seconds corresponding to
$K = 6.25\times 10^4,\, 1.25\times 10^5,\, 2.5\times 10^5$ sample points, respectively). 

As noted in Section~\ref{section:connection}, the computational complexity of PLSO is $O(J^2K)$, where as the computational complexity of GP-PS is $O(N^2K)$. Since $N$, the number of samples per window, is fixed ($2,500$ samples), we expect both PLSO and GP-PS to be linear in terms of the number of samples $K$. Table~\ref{table:efficiency} indeed confirms that this is the case. However, we observe that PLSO is much more efficient than GP-PS.
\begin{table}[!ht]
	\caption{Runtime (s) for PLSO and GP-PS for varying length }
	\label{table:efficiency}
	\centering
	\begin{tabular}[b]{|c|c|c|}
		\hline
		& \textbf{ PLSO}& \textbf{GP-PS}\\
		\hline
		$T=$50 & 1.7 & 346.8\\
		$T=$100 & 3.1 & 700.6\\
		$T=$200& 6.5 & 1334.0\\
		\hline
	\end{tabular}
\end{table}

\section*{J. Simulation experiment \textit{(Section 5.1)}}
We simulate from the following model for $1\leq k \leq K$
\begin{equation*}
\y_k = 10\left(\frac{K-k}{K} \right)\realz_{1,k} + 10\cos^4(2\pi \omega_0 k)\realz_{2,k}+\nu_k,
\end{equation*}
where $\z_{1,k}$ and $\z_{2,k}$ are from the PLSO stationary generative model, i.e., $\sigma_{j,m}^2=\sigma_j^2$. The parameters are $\omega_0 /\omega_1/\omega_2=0.04/1/10$ Hz, $f_s=200$ Hz, $T=100$ seconds, $l_1=l_2=1$, and $\nu_k\sim\N(0,25)$. This stationary process comprises two amplitude-modulated oscillations, namely one modulated by a slow-frequency ($\omega_0=0.04$ Hz) sinusoid and the other a linearly-increasing signal~\cite{Ba14}. We assume a 2-second PS interval. For PLSO, we use $J=2$ components and 5 block coordinate descent iterations for optimizing $\theta$ and $\sigset$.

\section*{K. Details of the TVAR model}
As explained in Section~\ref{section:connection}, the TVAR model is defined as 
\begin{equation*}
\y_k=\sum_{p=1}^P a_{p,k}\y_{k-p}+\varepsilon_k,
\end{equation*}	
which can alternatively be written as
\begin{equation*}
\begin{split}
&\begin{pmatrix}
\y_k\\
\vdots\\
\y_{k-P+1}\\
\end{pmatrix}\\
&=\underbrace{\begin{pmatrix}
a_{1,k}& a_{2,k} &\cdots & a_{P-1,k}&a_{P,k}\\
1 & 0 & \cdots & 0 & 0\\
0 & 1 & \cdots & 0 & 0\\
\vdots & & & &\vdots\\
& \vdots & &\\
0 & 0 & \cdots & 1 & 0\\
\end{pmatrix}}_{A_k}\begin{pmatrix}
\y_{k-1}\\
\vdots\\
\y_{k-P}\\
\end{pmatrix}+\varepsilon_k.\\
\end{split}
\end{equation*}	
It is the transition matrix $A_k$ that determines the oscillatory component profile at time $k$, such as the number of components and the center frequencies. Specifically, $\{A_k\}_k$ are first fit to the data $\y$ and then eigen-decomposition is performed on each of the estimated $\{A_k\}_k$. More in-depth technical details can be found in \citet{West1999}.

We use publicly available code for the TVAR implementation\footnote{https://www2.stat.duke.edu/~mw/mwsoftware/TVAR/index.html}. We use TVAR order of $p=70$, as the models with lower orders than this value did not capture the theta-band signal - The lowest frequency band in these cases were the gamma band ($>30$ Hz). Even with the higher orders of $p$, and various hyperparameter combinations, we observed that the slow and theta frequency band was still explained by a single oscillatory component. For the discount factor, we used $\beta=0.999$ to ensure that the TVAR coefficients and, consequently, the decomposed oscillatory components do not fluctuate much.

\section*{L. Anesthesia EEG dataset \textit{(Section 5.3)}}
We show spectral analysis results for the EEG data of a subject anesthetized with propofol (This is a different subject from the main text.) The data last $T=2{,}500$ seconds, sampled at $f_s=250$ Hz. We assume a 4-second PS interval, use $J=9$ components and 5 block coordinate descent iterations for optimizing $\theta$ and $\sigset$.

Fig.~\ref{fig:anesthesia_supp} shows the STFT and the PLSO-estimated spectrograms. As noted in the main text, PLSO with stationarity assumption is too restrictive and fails to capture the time-varying spectral pattern. Both PLSO with $\lambda=0$ and $\lambda=\lambda_{\operatorname{CV}}$ are more effective in capturing such patterns, with the latter able to remove the artifacts and better recover the smoother dynamics.

\begin{figure}[!htb]
	\centering
	\includegraphics[width=	0.9\linewidth]{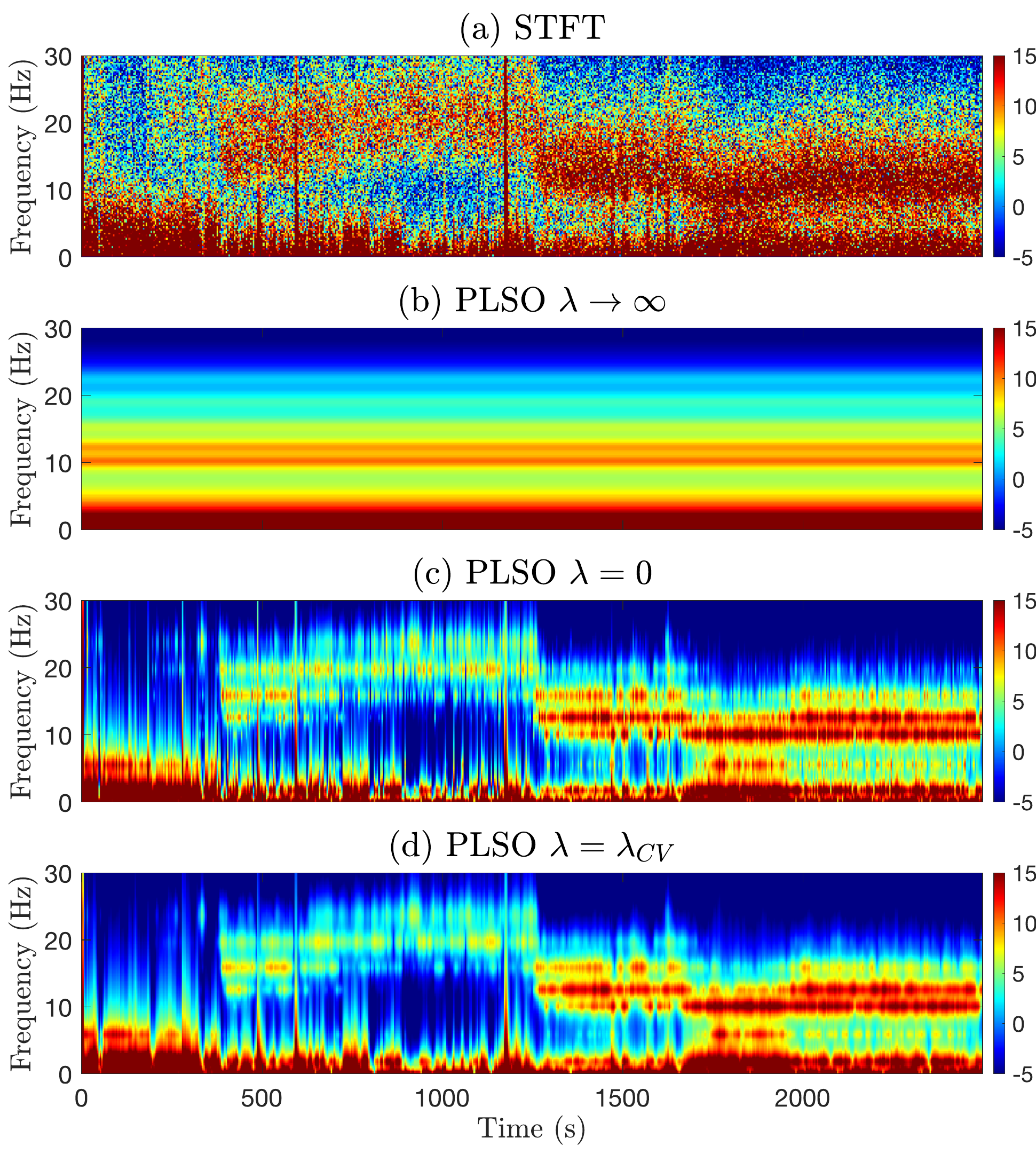}
	\caption{Spectrogram (in dB) under propofol anesthesia. (a) STFT of the data (b) PLSO with $\lambda\rightarrow \infty$ (c) PLSO with $\lambda=0$ (d) PLSO with $\lambda=\lambda_{\operatorname{CV}}$.}
	\label{fig:anesthesia_supp}
\end{figure}

\end{document}